\documentclass[11pt]{article}

\usepackage[font={small}]{caption}
\usepackage[margin=2.5cm]{geometry}
\usepackage[utf8x,utf8]{inputenc} 
\usepackage[T1]{fontenc}
\usepackage{tikz}
\usepackage{xspace}
\usepackage{amsmath,amssymb,amsthm}
\usepackage{enumerate}
\usepackage{url}
\usepackage{hyperref}
\usepackage{xcolor}
\usepackage{listings}

\definecolor{dark-red}{rgb}{0.4,0.15,0.15}
\definecolor{dark-blue}{rgb}{0.15,0.15,0.4}
\definecolor{medium-blue}{rgb}{0,0,0.5}
\definecolor{gray}{rgb}{0.5,0.5,0.5}

\hypersetup{
    colorlinks, linkcolor={dark-red},
    citecolor={dark-blue}, urlcolor={medium-blue}
}

\title{Fine-Grained Complexity of $k$-OPT in Bounded-Degree Graphs \\ for Solving TSP\footnote{This research was initiated at the Shonan Meeting {\em Parameterized Graph Algorithms \& Data Reduction: Theory Meets Practice} held  during March 4-8, 2019 in Shonan Village Center, Japan.
Yoichi Iwata thanks the Kaggle Traveling Santa 2018 Competition for motivating him to study practical TSP heuristics. He also thanks Shogo Murai for valuable discussion about the possibility of faster $k$-OPT algorithms.
}}

\newtheorem{theorem}{Theorem}[]
\newtheorem{lemma}[theorem]{Lemma}
\newtheorem{corollary}[theorem]{Corollary}

\newtheorem{observation}[theorem]{Observation}
\newtheorem{claim}{Claim}

\let\plainqed\qedsymbol
\newcommand{\claimqed}{$\lrcorner$}
\newenvironment{claimproof}{\begin{proof}\renewcommand{\qedsymbol}{\claimqed}}{\end{proof}\renewcommand{\qedsymbol}{\plainqed}}

\usepackage{xspace}
\usepackage{todonotes}
\usepackage{subcaption}
\usepackage[capitalize]{cleveref}
\usepackage{algorithmicx,algorithm}
\usepackage[noend]{algpseudocode}

\makeatletter
\renewcommand{\ALG@name}{Pseudocode}
\makeatother

\graphicspath{{./Figures/}}

\usepackage{comment}
\usepackage{thmtools}
\usepackage{thm-restate}
\newtheorem{assumption}{Assumption}
\newcommand{\ignore}[1]{}%

\newcommand{\ProblemFormat}[1]{{\sc #1}}
\newcommand{\ProblemName}[1]{\ProblemFormat{#1}\xspace}

\newcommand{\kOPT}{\textsc{$k$-opt}\xspace}
\newcommand{\kprimeOPT}{\textsc{$k'$-opt}\xspace}
\newcommand{\gain}{{\rm gain}}
\newcommand{\tw}{{\rm tw}}
\newcommand{\pw}{{\rm pw}}

\newcommand{\probKOPTOpt}{\ProblemName{$k$-opt Optimization}}
\newcommand{\probKOPTDec}{\ProblemName{$k$-opt Detection}}
\newcommand{\probSubIso}{\textsc{$k$-Partitioned Subgraph Isomorphism}\xspace}

\newcommand{\probOPTDec}[1]{\ProblemName{$#1$-opt Detection}}
\newcommand{\probTriangle}{\ProblemName{Triangle Detection}}

\newcommand{\heading}[1]{\medskip\noindent{\bf #1.\ }}%

\newcommand{\Oh}{{\ensuremath{\mathcal{O}}}}

\newcommand{\Cc}{{\ensuremath{\mathcal{C}}}}

\newcommand{\Tt}{{\ensuremath{\mathcal{T}}}}

\usetikzlibrary{fit}
\usetikzlibrary{intersections}
\usetikzlibrary{decorations.pathreplacing}
\usetikzlibrary{calc}

\tikzstyle{vp}=[circle,draw,inner sep=0.08pt, minimum size=0.15cm]

\def\bendedge{50}

\def\colo{red}
\def\colob{blue}

\DeclareMathOperator{\polylog}{polylog}

\newcommand{\defparproblem}[4]{
 \vspace{1mm}
\noindent\fbox{
 \begin{minipage}{0.96\textwidth}
 \begin{tabular*}{\textwidth}{@{\extracolsep{\fill}}lr} #1 & {\bf{Parameter:}} #3 \\ \end{tabular*}
 {\bf{Input:}} #2 \\
 {\bf{Question:}} #4
 \end{minipage}
 }
 \vspace{1mm}
}

\begin{document}

\date{}

\author{
  {\'E}douard Bonnet\thanks{ENS Lyon, LIP, Lyon, France, \texttt{edouard.bonnet@ens-lyon.fr}}
\and
  Yoichi Iwata\thanks{National Institute of Informatics, Japan, \texttt{yiwata@nii.ac.jp}. Supported by JSPS KAKENHI Grant Number JP17K12643.}
\and
  Bart M. P. Jansen\thanks{Eindhoven University of Technology, Eindhoven, The Netherlands, \texttt{b.m.p.jansen@tue.nl}. Supported by ERC Starting Grant ReduceSearch (Grant Agreement No 803421).}
\and
  \L ukasz Kowalik\thanks{Institute of Informatics, University of Warsaw, Poland, \texttt{kowalik@mimuw.edu.pl}. Supported by ERC Starting Grant TOTAL (Grant Agreement No 677651).}
}

\maketitle

\begin{abstract}
The \textsc{Traveling Salesman Problem} asks to find a minimum-weight Hamiltonian cycle in an edge-weighted complete graph. Local search is a widely-employed strategy for finding good solutions to TSP. A popular neighborhood operator for local search is $k$-opt, which turns a Hamiltonian cycle~$\mathcal{C}$ into a new Hamiltonian cycle~$\mathcal{C}'$ by replacing~$k$ edges. We analyze the problem of determining whether the weight of a given cycle can be decreased by a $k$-opt move. Earlier work has shown that (i) assuming the Exponential Time Hypothesis, there is no algorithm that can detect whether or not a given Hamiltonian cycle~$\mathcal{C}$ in an $n$-vertex input can be improved by a $k$-opt move in time~$f(k) n^{o(k / \log k)}$ for any function~$f$, while (ii) it is possible to improve on the brute-force running time of~$\Oh(n^k)$ and save linear factors in the exponent. Modern TSP heuristics are very successful at identifying the \emph{most promising} edges to be used in $k$-opt moves, and experiments show that very good global solutions can already be reached using only the top-$\mathcal{O}(1)$ most promising edges incident to each vertex. This leads to the following question: can improving $k$-opt moves be found efficiently in graphs of bounded degree? We answer this question in various regimes, presenting new algorithms and conditional lower bounds. We show that the aforementioned ETH lower bound also holds for graphs of maximum degree three, but that in bounded-degree graphs the best improving $k$-move can be found in time $\Oh(n^{(23/135+\epsilon_k)k})$, where $\lim_{k\rightarrow\infty} \epsilon_k = 0$. This improves upon the best-known bounds for general graphs. Due to its practical importance, we devote special attention to the range of~$k$ in which improving $k$-moves in bounded-degree graphs can be found in \emph{quasi-linear} time. For~$k \leq 7$, we give quasi-linear time algorithms for general weights. For~$k=8$ we obtain a quasi-linear time algorithm when the weights are bounded by $\Oh(\polylog n)$. On the other hand, based on established fine-grained complexity hypotheses about the impossibility of detecting a triangle in edge-linear time, we prove that the~$k = 9$ case does not admit quasi-linear time algorithms. Hence we fully characterize the values of~$k$ for which quasi-linear time algorithms exist for polylogarithmic weights on bounded-degree graphs.

In order to get the time bound of our algorithm for arbitrary values of $k$, we show that every $k$-edge even multigraph has pathwidth at most $23k/135+o(k)$. 
This result may be of independent interest.
Indeed, we observe that it yields improved running time bounds for counting $k$-vertex paths and cycles.
\end{abstract}

\clearpage

\section{Introduction}
\subsection{Motivation}
The \textsc{Traveling Salesman Problem} (TSP) hardly needs an introduction; it is one of the most important problems in combinatorial optimization, which asks to find a Hamiltonian cycle of minimum weight in an edge-weighted complete graph. Local search is widely used in practical TSP solvers~\cite{JohnsonM02,JohnsonMc97}.
The most commonly used neighborhood is a \emph{$k$-move} (or $k$-opt move).
A $k$-move on a Hamiltonian cycle $\mathcal{C}$ is a pair $(E^-,E^+)$ of  edge sets such that $E^-\subseteq E(\mathcal{C})$, $|E^-|=|E^+|=k$ and $(\mathcal{C}\setminus E^-)\cup E^+$ is also a Hamiltonian cycle.
Marx~\cite{Marx08} showed that finding an improving $k$-move (i.e., a $k$-move that results in a lighter Hamiltonian cycle) is W[1]-hard parameterized by $k$,
and this result was refined by Guo~et~al.~\cite{GuoHNS13} to obtain an $f(k)n^{\Omega(k/\log k)}$ lower bound under the Exponential Time Hypothesis (ETH).
For small values of $k$, the current fastest running time is $\Oh(n^k)$ for $k=2,3$ (by exhaustive search), $\Oh(n^3)$ for $k=4$~\cite{BergBJW16}, and $\Oh(n^{3.4})$ for $k=5$~\cite{DBLP:conf/esa/CyganKS17}.
Moreover, de~Berg~et~al.~\cite{BergBJW16} and Cygan~et~al.~\cite{DBLP:conf/esa/CyganKS17} showed that improving the running time to $\Oh(n^{3-\epsilon})$ for $k=3$ or $k=4$ implies a breakthrough result of an $\Oh(n^{3-\delta})$-time algorithm for \ProblemName{All-Pairs Shortest Paths}.

From the hardness shown by the theoretical studies, it seems that local search can be applied only to small graphs.
Nevertheless, state-of-the-art local search TSP solvers can deal with large graphs with tens of thousands of vertices.
This is mainly due to the following two heuristics.
\begin{enumerate}
	\item They sparsify the input graph by picking the top-$d$ important incident edges for each vertex according to an appropriate importance measure.
		  For example, Lin-Kernighan~\cite{LinK73} picks the top-5 nearest neighbors,
		  and its extension LKH~\cite{Helsgaun00} picks the top-5 $\alpha$-nearest neighbors, where the $\alpha$-distance of an edge is the increase of the Held-Karp lower bound~\cite{HeldK71} by including the edge.
		  The empirical evaluation by Helsgaun~\cite{Helsgaun00} showed that the sparsification by the $\alpha$-nearest neighbors can preserve almost optimal solutions.
	\item They mainly focus on \emph{sequential} $k$-moves. In general, $E^-\cup E^+$ is a set of edge-disjoint closed walks, each of which alternately uses edges in $E^-$ and $E^+$.
		  If it consists of a single closed walk, the move is called sequential.
		  Graphs of maximum degree $d$ with~$n$ vertices have at most $n(2(d-2))^{k-1}$ sequential $k$-moves
		  ($n$ choices for the starting point, 2 choices for the next edge in $E^-$, and at most $d-2$ choices for the next edge in $E^+$),
		  which is linear in $n$ when considering $d$ and $k$ as constants.
		  On the other hand, linear-time computation of non-sequential $k$-moves appears non-trivial.
		  Lin-Kernighan does not search for non-sequential moves at all, and after it finds a local optimum, it applies special non-sequential 4-moves called \emph{double bridges} to get out of the local optimum.
		  LKH-2~\cite{Helsgaun09} improves Lin-Kernighan by heuristically searching for non-sequential moves during the local search.
\end{enumerate}

This state of affairs raises the following questions: what is the complexity of finding improving $k$-moves in bounded-degree graphs? How does the complexity scale with~$k$, and can it be done efficiently for small values of~$k$? Since improving \emph{sequential moves} can be found in linear time for fixed~$k$ and~$d$, to answer these questions we have to investigate non-sequential $k$-moves in bounded-degree graphs.

\subsection{Our contributions}

We classify the complexity of finding improving $k$-moves in bounded-degree graphs in various regimes. We present improved algorithms that exploit the degree restrictions using the structure of $k$-moves, treewidth bounds, color-coding, and suitable data structures. We also give new lower bounds based on the Exponential Time Hypothesis (ETH) and hypotheses from fine-grained complexity concerning the complexity of detecting triangles. To state our results in more detail, we first introduce the two problem variants we consider; a weak variant to which our lower bounds already apply, and a harder variant which can be solved by our algorithms.

\defparproblem{\probKOPTDec}
{An undirected graph~$G$, a weight function~$w \colon E(G) \to \mathbb{Z}$, an integer~$k$, and a Hamiltonian cycle~$\mathcal{C} \subseteq E(G)$.}
{$k.$}
{Can~$\mathcal{C}$ be changed into a Hamiltonian cycle of strictly smaller weight by a $k$-move?}

The related optimization problem \probKOPTOpt is to compute, given a Hamiltonian cycle in the graph, a $k$-move that gives the largest cost improvement, or report that no improving $k$-move exists. With this terminology, we describe our results.

We show that \probKOPTDec is unlikely to be fixed-parameter tractable on bounded-degree graphs: we give a new constant-degree lower-bound construction to show that there is no function~$f$ for which \probKOPTDec on subcubic graphs with weights~$\{1,2\}$ can be solved in time~$f(k) \cdot n^{o(k / \log k)}$, unless ETH fails. Hence the running time lower bound for general graphs by Guo et al.~\cite{GuoHNS13} continues to hold in this very restricted setting. While the degree restriction does not make the problem fixed-parameter tractable, it is possible to obtain faster algorithms. By adapting the approach of Cygan et al.~\cite{DBLP:conf/esa/CyganKS17}, exploiting the fact that the number of sequential moves is linear in~$n$ in bounded-degree graphs, and proving a new upper bound on the pathwidth of an $k$-edge even graph, we show that \probKOPTOpt in $n$-vertex graphs of maximum degree $\Oh(1)$ can be solved in time $\Oh(n^{(23/135+\epsilon_k)k})=\Oh(n^{(0.1704+\epsilon_k)k})$, where $\lim_{k\rightarrow\infty} \epsilon_k = 0$. This improves on the behavior for general graphs, where the current-best running time~\cite{DBLP:conf/esa/CyganKS17} is~$\Oh(n^{(1/4 + \epsilon_k)k})$.

Since quasi-linear running times are most useful for dealing with large inputs, we perform a fine-grained analysis of the range of~$k$ for which improving $k$-moves can be found in time~$\Oh(n \polylog n)$ on $n$-vertex graphs. Observe that in the bounded-degree setting, the number of edges~$m$ is~$\Oh(n)$. We prove lower bounds using the hypothesis that detecting a triangle in an unweighted graph cannot be done in nearly-linear time in the number of edges~$m$, which was formulated in several ways by Abboud and Vassilevska Williams~\cite[Conjectures 2--3]{AbboudW14}. By an efficient reduction from \probTriangle, we show that an algorithm with running time~$\Oh(n \polylog n)$ for \probOPTDec{9} in subcubic graphs with weights~$\{1,2\}$ implies that a triangle in an $m$-edge graph can be found in time~$\Oh(m \polylog m)$, contradicting popular conjectures. We complement these lower bounds by quasi-linear algorithms for all~$k \leq 8$ to obtain a complete dichotomy for the case of integer weights bounded by $\Oh(\polylog n)$. When the weights are not bounded, we obtain quasi-linear time algorithms for all~$k \leq 7$, leaving open only the case~$k=8$.

In order to get the time bound of our algorithm for arbitrary values of $k$, we show that every $k$-edge even multigraph has pathwidth at most $23k/135+o(k)$. 
This result may be of independent interest.
Indeed, we observe that it improves the running time of the $k^{O(k)} n^{0.174k+o(k)}$-time algorithm of Curticapean, Dell and Marx~\cite{CurticapeanHomomorphsims} for counting $k$-vertex paths or $k$-vertex cycles in a given $n$-vertex graph. 
The new bound is $k^{O(k)} n^{23k/135+o(k)} = k^{O(k)} n^{0.1704k+o(k)}$ and it stays the same for counting subgraphs isomorphic to any fixed graph $H$ with all but $O(1)$ vertices of even degree.

\subsection{Organization}
Preliminaries are presented in Section~\ref{sec:preliminaries}. In Section~\ref{sec:xp-algorithms} we give faster XP algorithms for varying~$k$, and in particular we show the new bound for pathwidth of even multigraphs. Next, in Section~\ref{sec:cnt-even} we apply the bound to counting paths, cycles and other graphs that are close to having all vertices of even degree. By refining some ideas of  Section~\ref{sec:xp-algorithms}, we give quasi-linear-time algorithms for~$k \leq 8$ in Section~\ref{sec:small-k-algorithms}. Section~\ref{sec:small-k-lowerbound} gives the reduction from \probTriangle to establish a superlinear lower bound on subcubic graphs for~$k=9$. In Section~\ref{sec:xp-lowerbound} we describe the lower bound for varying~$k$.

\section{Preliminaries} \label{sec:preliminaries}
Given a graph $G$ edge-weighted by $w \colon E(G) \rightarrow \mathbb Z$, and a subset $F \subseteq E(G)$ of its edges, $w(F) := \sum_{e \in F} w(e)$.  
A \emph{$k$-move} on a Hamiltonian cycle $\mathcal C$ is pair $(E^-,E^+)$ of edge sets such that $|E^-|=|E^+|=k$ and $(\mathcal C \setminus E^-) \cup E^+$ is also a Hamiltonian cycle.
A $k$-move is called \emph{improving} if $w((\mathcal C \setminus E^-) \cup E^+) < w(\mathcal C)$, or equivalently and more simply $w(E^+) < w(E^-)$.  
A necessary condition for a pair $(E^-,E^+)$ to be a $k$-move is that the multiset of endpoints of $E^-$ is equal to the multiset of endpoints of $E^+$.
An exchange $(E^-,E^+)$ that satisfies this condition is called a \emph{$k$-swap}.
We say that a $k$-swap \emph{results} in the graph $(\mathcal C \setminus E^-) \cup E^+$.
Note that a $k$-swap always results in a spanning disjoint union of cycles.
A $k$-swap resulting in a graph with a single connected component is therefore a $k$-move.
An \emph{infeasible} $k$-swap is a $k$-swap which is not a $k$-move.

We say that a $k$-swap $(E^-,E^+)$ \emph{induces} the graph $E^- \cup E^+$.
As a slight abuse of notation, a $k$-swap will sometimes directly refer to this graph.
A $k$-swap $(E^-,E^+)$ such that all edges $E^- \cup E^+$ are visited by a single closed walk alternating between $E^-$ and $E^+$ is called \emph{sequential}.
In particular, in a simple graph, every $2$-swap is sequential.
One can notice that an infeasible (sequential) $2$-swap results in a disjoint union of exactly two cycles.
A $k$-move can always be decomposed into sequential $k_i$-swaps (with $\sum k_i = k$) but some $k$-moves cannot be decomposed into sequential $k_i$-moves. 
%Observe that on subcubic graphs, a $k$-swap necessarily induces a disjoint union of cycles.
%A $k$-swap $(E^-,E^+)$ is called \emph{$w$-weighted} if $w = w(E^+) - w(E^-)$.
The quantity $w(E^-) - w(E^+)$ is called the \emph{gain} of the swap $(E^-,E^+)$.
We distinguish \emph{neutral} swaps, with gain 0, \emph{improving} swaps, with strictly positive gain, and \emph{worsening} swaps, with strictly negative gain. 

For an integer $n$, we denote $[n]=\{1,\ldots, n\}$.
A {\em $k$-embedding} (or shortly: {\em embedding}) is an increasing function $f\colon [k]\rightarrow [n]$. 
A {\em connection $k$-pattern} (or shortly: {\em connection pattern}) is a perfect matching in the complete graph on the vertex set $[2k]$.
%We call a connection pattern {\em valid} when one obtains a single $k$-cycle from $M$ by identifying vertex $2i$ with vertex $(2i+1)\bmod 2k$ for every $i=1,\ldots,k$.
A pair $(f,M)$ where $f$ is a $k$-embedding and $M$ is a connection $k$-pattern, is an alternative description of a $k$-swap.
Indeed, let $e_1,\ldots,e_n$ be subsequent edges of $\mathcal{C}$.
Then, $E^-=\{e_{f(i)} \colon i\in[k]\}$.
Vertices of the connection pattern correspond to endpoints of $E^-$, i.e., vertices $2i-1,2i\in[2k]$ correspond to the left and right (in the clockwise order) endpoint of $e_{f(i)}$, respectively.
Thus, edges of the connection pattern correspond to a set $E^+$ of $|M|$ edges in $G$.
We say that a $k$-swap $(E^-,E^+)$ \emph{fits into $M$} if there is an embedding $f$ such that $(f,M)$ describes $(E^-,E^+)$.
Note that every pair of an embedding and a connection pattern $(f, M)$ describes exactly one swap $(E^-,E^+)$. 
Conversely, for a swap $(E^-,E^+)$ the corresponding embedding $f$ is also unique (and determined by $E^-$).
However, in case $E^-$ contains incident edges, the swap fits into more than one matching $M$ (see Fig.~\ref{fig:sequential}).  
See~\cite{DBLP:conf/esa/CyganKS17} for a more formal description of the equivalence.
The main advantage of using connection patterns is that the feasibility and sequentiality of a $k$-swap $(f,M)$ can be determined only from $M$.

The notion of a connection pattern can be extended to represent $k'$-swaps, for $k'<k$, as follows.
Note that a matching $N$ in the complete graph on the vertex set $[2k]$ corresponds to an $|N|$-swap if and only if there is a set $\iota(N)\subseteq[k]$ such that $V(N)=\{2i-1,2i \colon i \in \iota(N)\}$.
For a set $X\subseteq[k]$, by $M[X]$ we denote the swap $N$ such that $\iota(N)=X$.
We say that a connection pattern $M$ {\em decomposes} into swaps $N_1,\ldots,N_t$ when $M=\biguplus_{i=1}^t N_i$ and each $N_i$ is a connection pattern of a swap.
The notion of fitting extends to $k'$-swaps in the natural way.

\begin{figure}[t]
\begin{center}
\begin{tabular}{c@{\quad\quad}c@{\quad\quad}c}
\begin{tikzpicture}[baseline,scale=0.55,line width=0.3mm]
\begin{scope}[>=latex]
 \tikzstyle{point}=[circle,fill=white,draw,minimum size=0.18cm,inner sep=0pt]

 \def\vert{10}
 \foreach \i in {1,...,\vert}
 {
   \node[point] (v\i) at (-90 + \i * 360 / \vert :25mm) {};
 }  

 {
 \foreach \i/\l in {1,...,\vert}
 {
   \node at ({-90 - (\i + 3.5) * 360 / \vert} :30mm) {$e_{\i}$};
 }  
 }

 {
 \foreach \i/\j in {1/2,2/3,3/4,6/7,7/8,8/9}
 {
       \draw [color=black, thick] (v\i) -- (v\j);
 }
 }

 {
  \foreach \i/\j in {1/10,9/10,5/6,4/5}
 {
       \draw [color=red,ultra thick] (v\i) -- (v\j);
 }
 }
 
 \foreach \i/\j in {4/10,5/1,5/9,6/10}
 {
       \draw [dashed,color=green!60!black,ultra thick] (v\i) -- (v\j);
 }
 \end{scope}
\end{tikzpicture}
&
\begin{tikzpicture}[baseline,scale=0.55,line width=0.3mm]
\begin{scope}[>=latex]
 \tikzstyle{point}=[circle,fill=white,draw,minimum size=0.18cm,inner sep=0pt]

 \def\vert{10}
 \foreach \i in {1,2,3,4,6,7,8,9}
 {
   \node[point] (v\i) at ({180- \i * 360 / 10} :25mm) {};
 }  

 {
 \foreach \i/\l in {1/1,2/2,3/3,4/4,6/5,7/6,8/7,9/8}
 {
   \node at ({180 - \i * 360 / 10} :30mm) {$\l$};
 }  
 }

 {
 \foreach \i/\j in {1/2,3/4,6/7,8/9}
 {
       \draw [color=red,ultra thick] (v\i) -- (v\j);
 }
 }

 \foreach \i/\j in {1/8,2/9,3/6,4/7}
 {
       \draw [dashed,color=green!60!black,ultra thick] (v\i) -- (v\j);
 }
 \end{scope}
\end{tikzpicture}
&
\begin{tikzpicture}[baseline,scale=0.55,line width=0.3mm]
\begin{scope}[>=latex]
 \tikzstyle{point}=[circle,fill=white,draw,minimum size=0.18cm,inner sep=0pt]

 \def\vert{10}
 \foreach \i in {1,2,3,4,6,7,8,9}
 {
   \node[point] (v\i) at ({180- \i * 360 / 10} :25mm) {};
 }  

 {
 \foreach \i/\l in {1/1,2/2,3/3,4/4,6/5,7/6,8/7,9/8}
 {
   \node at ({180 - \i * 360 / 10} :30mm) {$\l$};
 }  
 }

 {
 \foreach \i/\j in {1/2,3/4,6/7,8/9}
 {
       \draw [color=red,ultra thick] (v\i) -- (v\j);
 }
 }

 \foreach \i/\j in {1/7,2/9,3/6,4/8}
 {
       \draw [dashed,color=green!60!black,ultra thick] (v\i) -- (v\j);
 }
 \end{scope}
\end{tikzpicture}
\end{tabular}
\end{center}
\caption{\label{fig:sequential}A sequential swap (left) which fits two connection patterns (center, right). The pattern in the center is not sequential, while the pattern on the right is sequential. In the left picture the solid red edges are in $E^-$ , the dashed green edges are in $E^+$, and the thin black edges are the remaining edges of the Hamiltonian cycle $\mathcal{C}$.
In the central and right pictures, the dashed green edges form some connection patterns.}
\end{figure}
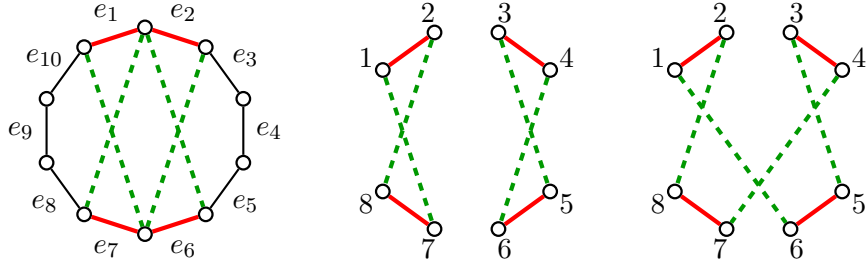

Consider a connection pattern $N$ of a swap, for $V(N)\subseteq [2k]$. We call $N$ {\em sequential} if $N \cup \{\{2i-1,2i\} \colon i \in \iota(N)\}$ forms a simple cycle. 
In particular, every connection pattern can be decomposed into sequential connection patterns of (possibly shorter) swaps.
The correspondence between sequential swaps and sequential connection patterns is somewhat delicate, so let us explain it in detail.

Let $N$ be a sequential connection pattern, $V(N)\subseteq [2k]$.
Recall that for every embedding $f$ there is exactly one $|N|$-swap $(E^-,E^+)$ that fits into $N$. Clearly, this swap is sequential, since every edge in $\{\{2i-1,2i\} \colon i \in \iota(N)\}$ corresponds to an edge of $E^-$ and every edge in $N$ corresponds to an edge in $E^+$.
Thus the resulting set of edges $E^-\cup E^+$ forms a single closed walk.
In particular, if the image of $f$ contains two neighboring indices $i, i+1\in[n]$, the closed walk is not a simple cycle.

Conversely, it is possible that a sequential swap fits into a connection pattern which is not sequential, see Fig.~\ref{fig:sequential} for an example.  
However, every sequential $\ell$-swap $(E^-,E^+)$ fits at least one sequential connection pattern.
This sequential connection pattern is determined by the closed walk which certifies the sequentiality of the swap.
Indeed, let $E^-=\{e_{i_1},\ldots,e_{i_{\ell}}\}$, where $i_1,\ldots,i_{\ell}$ is an increasing sequence.
Let $v_0,\ldots,v_{2\ell-1}$ be the closed walk alternating between $E^-$ and $E^+$, in particular assume that $E^-=\{v_iv_{i+1} \colon \text{$i$ is even}\}$. 
Consider any $i=0,\ldots,\ell-1$ and the corresponding edge $e_{i_j}=v_{2i}v_{2i+1}$ in $E^-$, for some $j\in[\ell]$. 
If $v_{2i}$ is the left endpoint of $e_{i_j}$, we put $w_{2i}=2j-1$ and $w_{2i+1}=2j$, otherwise $w_{2i}=2j$ and $w_{2i+1}=2j-1$.
Then $w_0,\ldots,w_{2\ell-1}$ is a simple cycle and $N=\{w_iw_{i+1} \colon \text{$i$ is odd}\}$ is a sequential connection pattern.
By construction, $(E^-,E^+)$ fits $N$, as required.
Keeping in mind the nuances in the notions of sequential swaps and corresponding sequential connection patterns, for simplicity, we will often just say `a sequential swap $M$' for a matching $M$, instead of the more formal `a sequential connection pattern $M$ of a swap'.

Fix a connection pattern $M$ and let $f\colon S\rightarrow[n]$ be a partial embedding, for some $S\subseteq [k]$.
For every $j\in S$, let $v_{2j-1}$ and $v_{2j}$ be the left and right endpoint of $e_{f(j)}$, respectively.
We define
\begin{align*}
E^-_f &= \{e_{f(i)} \mid i\in S\},\\
E^+_f &= \{\{v_{i'},v_{j'}\} \mid i,j\in S, i'\in \{2i-1,2i\}, j'\in \{2j-1,2j\}, \{i',j'\}\in M\}.
\end{align*}
Then, $\gain_M(f)=w(E^-_f)-w(E^+_f)$.

\section{Fast XP algorithms} \label{sec:xp-algorithms}
For all fixed integers $k$ and $d$, the number of sequential $k$-swaps in a graph of maximum degree $d$ is $\Oh(n)$, and we can enumerate all of them in the same running time.
Therefore, we can find the best improving $k$-move that can be decomposed into at most $c$ sequential $k$-swaps in $\Oh(n^c)$ time.
Because $c$ is at most $\lfloor\frac{k}{2}\rfloor$, we obtain an $\Oh(n^{\lfloor\frac{k}{2}\rfloor})$-time algorithm for \probKOPTOpt.
In what follows, we will improve this naive algorithm.
Below we present a relatively simple algorithm which exploits the range tree data structure~\cite{PreparataS85} and achieves running time roughly the same as the more sophisticated algorithm of Cygan et al.~\cite{DBLP:conf/esa/CyganKS17} for general graphs.

\begin{theorem}\label{thm:small-k-c}
For all fixed integers $k$, $c$, and $d$, there is an $\Oh(n^{\lceil\frac{c}{2}\rceil}\polylog{n})$-time algorithm to
compute the best improving $k$-move that can be decomposed into $c$ sequential swaps in graphs of maximum degree $d$.
\end{theorem}
\begin{proof}
When $c=1$, we can use the naive algorithm.
Suppose $c\geq 2$ and let $h:=\lceil\frac{c}{2}\rceil$.

For each possible connection pattern $M$ consisting of $c$ sequential swaps, we find the best embedding as follows.
Let $M=\bigcup_{i=1}^cN_i$, where each $N_i$ corresponds to a sequential swap.
We split $M$ into two parts $M_L=\bigcup_{i=1}^hN_i$ and $M_R=\bigcup_{i=h+1}^cN_i$ and we define $L=\bigcup_{i=1}^h \iota(N_i)$ and $R=\bigcup_{i=h+1}^c \iota(N_i)$.
Note that $L\uplus R=[k]$.
Let $f_L \colon L\to [n]$ and $f_R \colon R\to [n]$ be embeddings of $L$ and $R$, respectively.
The union of these two embeddings results in an embedding of $[k]$ if and only if the following conditions hold.
\begin{itemize}
	\item For each $i\in[k-1]$ with $i\in L$ and $i+1\in R$, $f_L(i)<f_R(i+1)$ holds.
	\item For each $i\in[k-1]$ with $i\in R$ and $i+1\in L$, $f_R(i)<f_L(i+1)$ holds.
\end{itemize}

We can efficiently compute a pair of embeddings satisfying these conditions using an orthogonal range maximum data structure as follows.
Let $\{l_1,\ldots,l_p\}=\{i \colon  \text{$l_i\in L$ and $l_i+1\in R$}\}$ and let $\{r_1,\ldots,r_q\}=\{i \colon  \text{$r_i-1\in R$ and $r_i\in L$}\}$.
We first enumerate all the $|L|$-swaps that fit into $M_L$ and all the $|R|$-swaps that fit into $M_R$, in $\Oh(n^h)$ time.
For each such $|L|$-swap $(f_L,M_L)$, we create a $(p+q)$-dimensional point $(f_L(l_1),\ldots,f_L(l_p),f_L(r_1),\ldots,f_L(r_q))$ with a priority $\gain_{M_L}(f_L)$, and we collect these points into a data structure. It stores $\Oh(n^h)$ points.
For each $|R|$-swap $(f_R,M_R)$, we query for the embedding $f_L$ of maximum priority satisfying $f_L(l_i)<f_R(l_i+1)$ for every $i\in[p]$ and $f_R(r_i-1)<f_L(r_i)$ for every $i\in[q]$,
and we answer the pair maximizing the total gain, i.e., the sum $\gain_{M_L}(f_L)+\gain_{M_R}(f_R)$.
Using the range tree data structure~\cite{PreparataS85}, each query takes $\Oh(\log^{p+q}{n^h})=\Oh(\polylog{n})$ time, so the total running time is $\Oh(n^h\polylog{n})$.
\end{proof}

Since $c\le \lfloor\frac{k}{2}\rfloor$ we get the following corollary.

\begin{corollary}
\label{cor:k/4}
  For all fixed integers $k$ and $d$, \probKOPTOpt in graphs of maximum degree $d$ can be solved in time $\Oh(n^{\lceil\frac{k-1}{4}\rceil}\polylog{n})$.
\end{corollary}

Let us take another look at the proof of Theorem~\ref{thm:small-k-c}.
Recall that for merging embeddings $f_L$ and $f_R$, we were interested only in values $f_L(i)$ for $i\in L$ such that $i+1\in R$ or $i-1\in R$. 
The embeddings of the remaining elements of $L$ were forgotten at that stage, but we knew that it is possible to embed them and we stored the gain of embedding them.
This suggests the following, different approach.

We decompose the connection pattern into sequential swaps and we scan the swaps in a carefully chosen order.
Assume we scanned $t$ swaps already and there are $c-t$ swaps ahead. 
Assume that only $p \ll t$ of the $t$ `boundary' swaps interact with the remaining $c-t$ swaps, where two swaps $N_1$ and $N_2$ interact when there is $i\in\iota(N_1)$ such that $i-1\in\iota(N_2)$ or $i+1\in\iota(N_2)$. 
Then it suffices to compute, for every embedding $f_L$ of the $p$ swaps, the gain of the best (i.e., giving the highest gain) embedding $g_L$ of the $t$ swaps, such that $f_L$ matches $g_L$ on the boundary swaps. This amounts to $\Oh(n^p)$ values to compute, since each sequential swap can be embedded in $\Oh(n)$ ways, if $k$ and the maximum degree are $\Oh(1)$.
The idea is to (1) compute these values quickly (in time linear in their number) using analogous values computed for the prefix of $t-1$ swaps, (2) find an order of swaps so that $p$ is always small, namely $p\le(23/135+\epsilon_k)k$.
The readers familiar with the notion of pathwidth recognize that $p$ here is just the pathwidth of the graph obtained from the path $1,2,\ldots,k$ by identifying vertices in the set $\iota(N)$ for every sequential swap $N$ in $M$, and that (2) is just dynamic programming over the path decomposition.
The resulting algorithm is summarized in Theorem~\ref{thm:fast-xp}.
%The resulting algorithm is summarized in Theorem~\ref{thm:fast-xp}, and due to space limits, its formal proof is moved to Appendix~\ref{sec:app-xp}.

\begin{theorem}
\label{thm:fast-xp}
  For all fixed integers $k$ and $d$, \probKOPTOpt in graphs of maximum degree $d$ can be solved in time $\Oh(n^{(23/135+\epsilon_k)k})=\Oh(n^{(0.1704+\epsilon_k)k})$, where $\lim_{k\rightarrow\infty} \epsilon_k = 0$.
\end{theorem}

%\todo[inline]{Edouard: It's a bit strange to fix $k$, and then to have this sequence $\epsilon_k$. Wouldn't it be more formal to say, there is sequence $\epsilon_k$ tending to 0, such that \probKOPTOpt on degree-d graphs can be solved in time...?}
Before proving Theorem~\ref{thm:fast-xp}, we first show the desired bound on the pathwidth.
We assume that the reader is familiar with the standard notions of treewidth, pathwidth and nice tree decompositions (see e.g.~\cite{fpttextbook}).

The following result is a more precise version of Lemma 1 from the work of Fomin et al.~\cite{fomin-algorithmica-2009}, and can be proved by iterating their arguments for larger degrees.

\begin{theorem}[Fomin et al.~\cite{fomin-algorithmica-2009}]
\label{thm:bound-tw-fomin}
For any $\epsilon > 0$, there exists an integer $n_\epsilon$ such that for every multigraph $G$ with $n>n_\epsilon$ vertices,
\[\pw(G) \le \frac{1}6 n_3 + \frac{1}3 n_4 + \frac{13}{30} n_5 + \frac {23}{45}n_6 + \frac {359}{630}n_7 + \frac{1553}{2520}n_8 + \frac{14827}{22680}n_9 + \frac{155273}{226800}n_{10} + n_{\ge 11} + \epsilon n,\]
where $n_i$ is the number of vertices of degree $i$ in $G$ for any $i \in \{3,\ldots, 10\}$ and $n_{\ge 11}$
is the number of vertices of degree at least 11. 
\end{theorem}

A graph is called {\em even} when all its vertices have even degrees.

\begin{lemma}
\label{lem:pw:even}
For any $\epsilon > 0$, there exists an integer $m_\epsilon$ such that for every even multigraph $G$ with $m>m_\epsilon$ edges, 
\[\pw(G) \le \frac{23}{135} m  + \epsilon m.\]
\end{lemma}

\begin{proof}
The bound follows from Theorem~\ref{thm:bound-tw-fomin} by solving the following linear program.
 \begin{equation}
\begin{aligned}
& \text{maximize}
& & \tfrac{1}{3}x_4 + \tfrac {23}{45}x_6 + \tfrac{1553}{2520}x_8 + \tfrac{155273}{226800}x_{10} + x_{12} & \\
& \text{subject to} & & 2x_4 + 3x_6 + 4x_8 + 5x_{10} + 6x_{12} \le 1 & \\
& & & x_i \ge 0 &\text{for $i=4,6,8,10,12$.}
\end{aligned}
\end{equation}
Indeed, let $n_i$ be the number of degree $i$ vertices for every  $i=2,4,6,8,10$ and let $n_{12}$ be the number of vertices of degree at least 12.
In the LP above variable $x_i$ corresponds to $n_i/m$ and the constraint is equivalent to $\sum_{i\ge 0} i \cdot n_i \le 2m$.
\end{proof}

We are now ready to prove Theorem~\ref{thm:fast-xp}.

\begin{proof}[Proof of Theorem~\ref{thm:fast-xp}]
The proof is inspired by the work of Cygan et al.~\cite{DBLP:conf/esa/CyganKS17}.
Our algorithm iterates over all at most $(2k)!$ connection patterns, so let us fix a connection pattern $M$, and let $M=\biguplus_{i=1}^cN_i$ be a decomposition of $M$ into sequential swaps.
Recall from Section~\ref{sec:small-k-algorithms} that we say that swaps $N_p$ and $N_q$ {\em interact} when there is $i\in [k]$ such that $i\in\iota(N_p)$ and $i-1\in\iota(N_q)$ or $i+1\in\iota(N_q)$.
We create a graph $C_M$ which describes interactions between the swaps.
The vertex set of $C_M$ is $[c]$, i.e., vertex $i\in [c]$ corresponds to the swap $N_i$.
We put an edge $ij$ whenever $N_i$ and $N_j$ interact.
Note that $C_M$ is a subgraph of the multigraph $\overline{C_M}$, obtained from the cycle $1,2,\ldots,k$ by identifying verties of $\iota(N_i)$, for every $i\in[c]$.
Multigraph $\overline{C_M}$ is even and has exactly $k$ edges.
Hence, by Lemma~\ref{lem:pw:even}, $\pw(C_M) \le \pw(\overline{C_M}) \le (\frac{23}{135} +  \epsilon_k) k$, where $\lim \epsilon_k = 0$.
We claim that an optimal embedding $f:[k]\rightarrow[n]$ can be computed by dynamic programming (DP) in time $\Oh(n^{\tw(C_M)+1})$ for every fixed $k,d$.

%The DP is similar to the one in~\cite{DBLP:conf/esa/CyganKS17}, so we give only a rough sketch.
For $S\subseteq V(C_M)$, denote $\widehat{S} = \bigcup_{j\in S} \iota(N_j)$.
Let $G$ be the input graph of maximum degree $d$.
We say that a partial embedding $f:X\rightarrow[n]$ is {\em admissible} when we have $E^+_f \subseteq E(G)$.

Let $\Tt=(T,\{X_t\}_{t\in V(T)})$ be a nice tree decomposition of $C_M$ with minimum width.
For every $t\in V(T)$ we denote by $V_t$ the union of all the bags in the subtree of $T$ rooted in $t$.
For every node $t\in V(T)$, and for every increasing and admissible function $f:\widehat{X_t}\rightarrow[n]$, we compute the following value.
\[T_t[f] = \max_{\substack{g:\widehat{V_t} \rightarrow [n]\\g|_{\widehat{X_t}}=f\\\text{$g$ is increasing}\\\text{$g$ is admissible}}}\gain_M(g).\]
Algorithms for processing introduce, forget and join nodes of $\Tt$ are essentially the same as in~\cite{DBLP:conf/esa/CyganKS17}, because the new DP is basically the old DP truncated only to admissible embeddings with a domain of the form $\widehat{S}$ for a set $S$. Therefore here we just state the formulas and skip their formal correctness proofs.

\heading{Leaf node}
When $t$ is a leaf of $T$, we know that $X_t=V_t=\emptyset$, and we just put $T_t[\emptyset]=0$.

\heading{Introduce node}
Assume $X_t = X_{t'}\cup\{i\}$, for some $i\not\in X_{t'}$ where node $t'$ is the only child of $t$.
Denote $\Delta E^+_f = E^+_f\setminus E^+_{f|_{\widehat{X_{t'}}}}$.
Then, for every admissible increasing function $f:\widehat{X_t}\rightarrow[n]$, 
\begin{equation}
\label{eq:intro}
T_t[f] = T_{t'}[f|_{\widehat{X_{t'}}}] + \sum_{j\in\iota(N_i)}w(e_{f(j)}) - \sum_{\{u,v\}\in \Delta E^+_f}w(\{u,v\}). 
\end{equation}

\heading{Forget node}
Assume $X_t = X_{t'}\setminus\{i\}$, for some $i\in X_{t'}$ where node $t'$ is the only child of $t$.
Then the definition of $T_t[f]$ implies that
\begin{equation}
\label{eq:forget}
T_t[f] = \max_{\substack{f':\widehat{X_{t'}} \rightarrow [n]\\f'|_{\widehat{X_t}}=f\\\text{$f'$ is increasing}\\\text{$f'$ is admissible}}} T_{t'}[f']. 
\end{equation}

\heading{Join node}
Assume $X_t = X_{t_1} = X_{t_2}$, for some nodes $t$, $t_1$ and $t_2$, where $t_1$ and $t_2$ are the only children of $t$.
Then, for every admissible increasing function $f:X_t\rightarrow[n]$, 
\begin{equation}
\label{eq:join}
T_t[f] = T_{t_1}[f] + T_{t_2}[f] - \left(w(E^-_f) - w(E^+_f)\right). 
\end{equation}

Clearly, in each kind of node we use time linear in the size of the DP tables.
The number of admissible increasing functions $f:\widehat{X_t}\rightarrow[n]$ is
just the number of embeddings of the sequential swaps $N_i$, $i\in X_t$, which amounts to $\Oh(n^{|X_t|})$, since a single sequential swap can be embedded in $\Oh(n)$ ways.
As $|X_t|\le \tw(C_M)+1$, the claimed running time bound of the whole algorithm follows.
\end{proof}

\section{A byproduct: counting even subgraphs} \label{sec:cnt-even}
Lemma~\ref{lem:pw:even} has an immediate consequence in improved running times of algorithms 
for counting paths and cycles in graphs and, more generally, other subgraphs that are {\em even} 
(i.e., with all vertices of even degree) or almost even. Let us recall the following theorem by 
Curticapean, Dell and Marx~\cite{CurticapeanHomomorphsims}.

\begin{theorem}{\cite{CurticapeanHomomorphsims}}
\label{thm:cnt}
   Given as input a $k$-edge graph~$H$ and an~$n$-vertex graph~$G$,
  we can compute the number of subgraphs in $G$ that are isomorphic to $H$ in time $k^{O(k)}\cdot n^{t+1}$, 
  where $t$ is the maximum treewidth in the spasm of $H$.
\end{theorem}

Here, a \emph{spasm of~$H$} is the set of all homomorphic images of~$H$,
that is, the set of ``all possible non-edge contractions'' of~$H$.
By applying the bound of Lemma~\ref{lem:pw:even} in Theorem~\ref{thm:cnt} we get the following result.

\begin{theorem}
   Given as input a $k$-edge even graph~$H$ and an~$n$-vertex graph~$G$,
  we can compute the number of subgraphs in $G$ that are isomorphic to $H$ in time $k^{O(k)}\cdot n^{23k/135 + o(k)} = k^{O(k)}\cdot n^{0.1704k + o(k)}$.
\end{theorem}

\begin{proof}
Consider a homomorphic image $\hat{H}$ of $G$ and the corresponding homomorphism $h:V(H)\rightarrow V(\hat{H})$.
We can define a multigraph $\tilde{H}$ on $V(\hat{H})$ with $k$ edges, where every edge $uv$ of $H$ corresponds to an edge $h(u)h(v)$ in $\tilde{H}$.
Then, $\tilde{H}$ is an even multigraph. 
Hence, by Lemma~\ref{lem:pw:even}, $\tilde{H}$ has pathwidth at most $23k/135 + o(k)$.
Since $\hat{H}$ is a subgraph of $\tilde{H}$, we get $\pw(\hat{H})\le 23k/135 + o(k)$.
The claim follows by Theorem~\ref{thm:cnt}.
\end{proof}

The following corollary is immediate (for paths just note that graphs in a spasm of a $k$-vertex path are subgraphs of graphs in the spasm of a $(k+1)$-vertex cycle).

\begin{corollary}
  We can compute the number of $k$-vertex paths or $k$-vertex cycles in a given $n$-vertex graph in time $k^{O(k)}\cdot n^{23k/135 + o(k)}= k^{O(k)}\cdot n^{0.1704k + o(k)}$.
\end{corollary}

\section{Fast algorithms for small $k$} \label{sec:small-k-algorithms}
Note that the algorithm for \probKOPTOpt from Corollary~\ref{cor:k/4} is quasi-linear for $k\leq 5$.
In this section we extend the quasi-linear-time solvability to $k\leq 7$ for \probKOPTDec.
Under an additional assumption of bounded weights, we are able to reach quasi-linear time for $k=8$ as well.
To be precise, we prove the following stronger statements than just finding an arbitrary improving $k$-move.

\begin{theorem}\label{thm:at-most-7}
For $k\leq 7$,
there is a quasi-linear-time algorithm to compute the best improving $k$-move in bounded-degree graphs under the assumption that there are no improving $k'$-moves for $k'<k$.
\end{theorem}

\begin{theorem}\label{thm:8opt}
When all the weights are integers from $\{1,2,\cdots,W\}$, there is an $\Oh(W^2 n \polylog n)$-time algorithm to compute the best improving $8$-move under the assumption that there are no improving $k'$-moves for $k'<8$.
\end{theorem}

We say that a connection pattern $M$ of $k$-swaps is \emph{reducible} if it can be decomposed into two moves.
Note that if $M$ is improving, then at least one of the two moves is improving, contradicting the assumption of the theorems.

\begin{observation}
If there are no improving $k'$-moves for $k'<k$, then no improving $k$-swap fits into a reducible connection pattern.
\end{observation}

% We will use the assumption that there are no improving $k'$-moves for $k'<k$ as follows.
% Suppose that an improving $k$-swap (which may be infeasible) can be decomposed into $k'$-move and $(k-k')$-move.
% Then at least one of the two moves is improving, which contradicts the assumption.
% Therefore, under our assumption, no improving $k$-swap can be decomposed into two moves.
% We say that a connection pattern $M$ of $k$-swaps \emph{reducible} if it can be decomposed into two moves.
% By the above argument, our assumption implies that there are no improving $k$-swaps fit into reducible connection patterns.
% We thus can focus only on irreducible connection patterns.
% In order to analyze the feasibility of 2-swaps, we will use the following observation.
% 
% \begin{observation}
% 2-swap $(E^-,E^+)$ on a cycle $\mathcal{C}$ is feasible if and only if, when walking along $E^-\cup E^+$, exactly one of the two edges in $E^-$ has the opposite direction to $\mathcal{C}$.
% \end{observation}

Before we formulate our algorithms, we need two lemmas.
%We can prove these lemmas by case analysis, and because of the space constraints, their proofs are in Appendix~\ref{sec:lemmas-small-k}.
Let $M[X]$ and $M[Y]$ be two swaps in a connection pattern $M$, for some disjoint $X,Y\subseteq[k]$.
{\em Interaction} between $M[X]$ and $M[Y]$ is any $i\in[k-1]$ such that $i\in X$ and $i+1\in Y$ or $i\in Y$ and $i+1\in X$.
%In what follows we often identify $X$ with $M[X]$, in particular we will speak about interactions between $X$ and $Y$, we will write $\mathcal{C}\oplus X$ etc.

% Algorithm~\ref{alg:k7} describes our algorithm.
% For each irreducible connection pattern $M$, we compute the best embedding as follows.
% If $M$ consists of at most two sequential swaps, we can use the algorithm in Theorem~\ref{thm:small-k-c}.
% Otherwise, $M$ consists of three sequential swaps.
% Let $[k]=X\uplus Y\uplus Z$.
% For $k\leq 7$, we can assume that $|X|=|Y|=2$ and $|Z|=k-4$.
% We use $M[S]$ to denote the $|S|$-swap for $S\subseteq [k]$.
% We define the number of \emph{interactions} between $X$ and $Y$ as the number of indices $i\in[k-1]$ such that $i\in X$ and $i+1\in Y$, or $i\in Y$ and $i+1\in X$.
% Note that there is no index $i\in [k-1]$ such that $X=\{i,i+1\}$ or $Y=\{i,i+1\}$ because in such a case, both the 2-swap and the remaining $(k-2)$-swap have to be feasible.
% The following lemma shows that our assumption implies that the number of interactions between $X$ and $Y$ is at most one.

\begin{lemma}\label{lem:interactions}
For any $k\ge 6$, there is no feasible and irreducible connection $k$-pattern that contains two 2-swaps that interact at least twice.
\end{lemma}

Let $M$ be a connection pattern, i.e., a perfect matching on vertices $[2k]$.
We say that $M'$ is obtained from $M$ by swapping $i$ and $i+1$, for $i\in [k]$, when $M'$ is obtained from $M$ by swapping the mates of $2i-1$ and $2i+1$ and swapping the mates of $2i$ and $2i+2$.

\begin{lemma}\label{lem:relax1}
Let $M$ be a feasible irreducible connection $k$-pattern.
Assume that $M$ decomposes into three sequential swaps $M[X]$, $M[Y]$, and $M[Z]$, such that $|X|=|Y|=2$.
If there is exactly one index $i\in [k-1]$ with $i\in X$ and $i+1\in Y$ or $i\in Y$ and $i+1\in X$,
the connection pattern $M'$ obtained from $M$ by swapping $i$ and $i+1$ is either feasible or reducible.
\end{lemma}

We prove Lemmas~\ref{lem:interactions} and~\ref{lem:relax1} by case analysis.
Alternatively, one can check the lemmas against a fixed $k$ by writing a program to generate all the connection $k$-patterns as we will do later for Lemma~\ref{lem:relax2}.
Proving the lemmas against $k\leq 8$ suffices for our purpose.
We say that a swap $S=(E^-,E^+)$ {\em connects} cycles $C_1$, $C_2$ when $E^-$ intersects both $C_1$ and $C_2$.

\begin{observation}
\label{obs:connect}
If a 2-swap $S$ connects two cycles, then after applying $S$ the two cycles are replaced by a single cycle, regardless of whether $S$ was feasible or not.
\end{observation}

\begin{observation}
Let~$M$ be a connection $k$-pattern that contains a 2-swap~$M[X]$ for some~$X \subseteq [k]$, such that~$M[X]$ is a matching on the vertex set~$\{2i-1, 2i \mid i \in X\}$. Let~$a<b<c<d$ be the endpoints of the two edges in~$M[X]$. Then the 2-swap $M[X]$ is feasible if and only if~$M[X]= \{ac,bd\}$, i.e., when the edges of~$M[X]$, seen as chords of the Hamiltonian cycle, cross each other.
%Let $N={i_1j_1,i_2j_2}$ be a connection pattern of a 2-swap, $i_1<i_2$. Then the 2-swap is feasible if and only if $i_1<i_2<j_1<j_2$, i.e., when edges of $N$, seen as chords of the Hamiltonian cycle, cross each other.
\end{observation}

Recall that $\mathcal{C}$ denotes the initial Hamiltonian cycle.
For a connection $k$-pattern~$M$, an index set~$X \subseteq [k]$ and the corresponding swap $M[X]$, we denote by $\mathcal{C}\oplus M[X]$ the cycle cover obtained by applying the swap $M[X]$ on $\mathcal{C}$.
%In what follows we often identify $X$ with $M[X]$, in particular we will speak about interactions between $X$ and $Y$, we will write $\mathcal{C}\oplus X$ etc.

% Algorithm~\ref{alg:k7} describes our algorithm.
% For each irreducible connection pattern $M$, we compute the best embedding as follows.
% If $M$ consists of at most two sequential swaps, we can use the algorithm in Theorem~\ref{thm:small-k-c}.
% Otherwise, $M$ consists of three sequential swaps.
% Let $[k]=X\uplus Y\uplus Z$.
% For $k\leq 7$, we can assume that $|X|=|Y|=2$ and $|Z|=k-4$.
% We use $M[S]$ to denote the $|S|$-swap for $S\subseteq [k]$.
% We define the number of \emph{interactions} between $X$ and $Y$ as the number of indices $i\in[k-1]$ such that $i\in X$ and $i+1\in Y$, or $i\in Y$ and $i+1\in X$.
% Note that there is no index $i\in [k-1]$ such that $X=\{i,i+1\}$ or $Y=\{i,i+1\}$ because in such a case, both the 2-swap and the remaining $(k-2)$-swap have to be feasible.
% The following lemma shows that our assumption implies that the number of interactions between $X$ and $Y$ is at most one.

\begin{proof}[Proof of Lemma~\ref{lem:interactions}]
Suppose that there is an irreducible feasible connection $k$-pattern $M$ such that for a triple $X\uplus Y\uplus Z=[k]$ we have that $M[X]$ and $M[Y]$ are 2-swaps that interact at least twice with~$\min(X) < \min(Y)$.
Note that, when $k\leq 7$, $M[Z]$ is a sequential $(k-4)$-swap, but in this lemma, we do not require that $M[Z]$ is sequential because we want to reuse this lemma later for $k=8$.

There are two ways in which $X$ and $Y$ can interact at least twice: (type 1) there is an index $i$ such that $i,i+2\in X$ and $i+1\in Y$, and (type 2) there are two indices $i$ and $j$ such that $X\cup Y=\{i,i+1,j,j+1\}$. In both types, $X\cup Y$ can be partitioned into two subsets of successive integers. 
It follows that the edges $E^-(M[X])\cup E^-(M[Y])$ intersect at most two cycles of $\Cc\oplus M[Z]$.
Hence, $\Cc\oplus M[Z]$ has at most two cycles, for otherwise $M$ is not feasible.

Assume $\Cc\oplus M[Z]$ has exactly two cycles.
If the interaction is of type 1, $M[X]$ cannot connect the cycles, and if the interaction is of type 2, both or neither $M[X]$ and $M[Y]$ connect the cycles.
At least one of $M[X]$ or $M[Y]$ has to connect the cycles, for otherwise $M$ is not feasible.
$M[Y]$ therefore connects the cycles, and then $M[Y\cup Z]$ is feasible by Observation~\ref{obs:connect}.
Because $M$ is irreducible, $M[X]$ is infeasible.
If the interaction is of type 1, no elements in $Z$ are between $i$ and $i+2$.
$X$ is therefore fully contained in one of the two cycles, let us denote it $\mathcal{C}_1$.
Because applying $M[Z]$ does not affect the edges of the Hamiltonian cycle between the embeddings of $i$ and $i+2$, $M[X]$ remains infeasible on $\mathcal{C}_1$.
Therefore, $M[X\cup Z]$ splits the Hamiltonian cycle into three cycles, and the 2-swap $M[Y]$ cannot connect three cycles into a single cycle, which is a contradiction.
If the interaction is of type 2, $M[X]$ also connects the two cycles, and therefore $M[Y]$ also has to be infeasible.
We then observe that, whether $M[X]$ and $M[Y]$ are parallel or cross, the entire $M$ is infeasible.
In more detail, if they are parallel (i.e., if $X=\{i,j+1\}$ and $Y=\{i+1,j\}$), we have a cycle with four vertices $\{2i,2i+1,2j,2j+1\}$ and a cycle containing the other vertices,
and if they cross (i.e., if $X=\{i,j\}$, and $Y=\{i+1,j+1\}$), we have a cycle containing $\{2i-1,2j,2j+1,2i+2\}$ and a cycle containing $\{2j-1,2i,2i+1,2j+2\}$.
These are two different cycles because if in $\Cc\oplus M$ there is a path from $2i+2$ to $2j-1$ or to $2j+2$, then $M[Y]$ does no connect two different cycles of $\Cc\oplus M[Z]$, a contradiction.

Assume $\Cc\oplus M[Z]$ has exactly one cycle.
Then, $M[Z]$ is feasible, so $M[X\cup Y]$ has to be infeasible, by irreducibility of $M$.
If the interaction is of type 1, $M[X\cup Y]$ is infeasible only when both $M[X]$ and $M[Y]$ are feasible.
Because applying $M[Z]$ does not affect the edges of the Hamiltonian cycle between the embeddings of $i$ and $i+2$, $M[X\cup Z]$ is also feasible.
Hence $M$ decomposes into two moves $M[X\cup Z]$ and $M[Y]$, which is a contradiction with irreducibility of $M$.
If the interaction is of type 2, because applying $M[Z]$ has to change the feasibility of $M[X\cup Y]$, it also changes the feasibility of both $M[X]$ and $M[Y]$.
If $M[X]$ is feasible and $M[Y]$ is infeasible, $M[Y\cup Z]$ becomes feasible, which is a contradiction.
We thus can assume that the feasibilities of $M[X]$ and $M[Y]$ are the same,
and in this case, $M[X\cup Y]$ is feasible if and only if (1) $M[X]$ and $M[Y]$ are parallel and both are feasible or (2) $M[X]$ and $M[Y]$ cross and both are infeasible.
Because applying $M[Z]$ has to flip the direction of exactly one of the two edges in $E^-$ of $M[X]$ (and $M[Y]$), it flips the states (1) and (2).
$M[X\cup Y]$ therefore remains infeasible after applying $M[Z]$, which is a contradiction.
\end{proof}

Let $M$ by be a connection pattern, i.e., a perfect matching on vertices $[2k]$.
We say that $M'$ is obtained from $M$ by swapping $i$ and $i+1$, for $i\in [k]$, when $M'$ is obtained from $M$ by swapping the mates of $2i-1$ and $2i+1$ and swapping the mates of $2i$ and $2i+2$.

\begin{proof}[Proof of Lemma~\ref{lem:relax1}]
Note that, in this lemma, we do not assume $k\leq 7$ because we want to reuse it later for $k=8$.
Because applying $M[X\cup Y]$ can connect at most three cycles, applying $M[Z]$ results in one, two, or three cycles.

Assume $\Cc\oplus M[Z]$ has exactly three cycles.
Then, $M[X]$ and $M[Y]$ connect distinct pairs of cycles.
Because swapping $i$ and $i+1$ does not affect the pairs, $M'$ is also feasible.

Assume $\Cc\oplus M[Z]$ has exactly two cycles.
Then, at least one of $M[X]$ and $M[Y]$ connects the two cycles, and let us assume $M[X]$ does so.
Applying $M[X\cup Z]$ then results in a single cycle and $M[Y]$ is feasible on this cycle.
After swapping $i$ and $i+1$ we again get one swap that connects the two cycles and one which is feasible on the resulting single cycle.
%Because swapping $i$ and $i+1$ cannot flip the direction of exactly one of the two edges in $E^-$ of $M[Y]$, $M'$ is also feasible.

Finally, assume $\Cc\oplus M[Z]$ has just one cycle, i.e., $M[Z]$ is feasible.
Then, $M[X\cup Y]$ has to be infeasible because $M$ is irreducible.
Therefore, we have either (1) $M[X]$ and $M[Y]$ cross and both are feasible, or (2) $M[X]$ and $M[Y]$ are parallel and at least one of them is infeasible.
Because swapping $i$ and $i+1$ flips the states of crossing and parallel but does not change the feasibility of $M[X]$ nor $M[Y]$,
$M'[X\cup Y]$ becomes feasible, and thus $M'$ is reducible.
\end{proof}

\begin{algorithm}[t!]
	\caption{Quasi-linear-time algorithm for $k\leq 7$}
	\label{alg:k7}
	\begin{algorithmic}[1]
	\For{each feasible irreducible connection $k$-pattern $M$}
		\If{$M$ consists of at most two sequential swaps}
			\State Apply the algorithm in Theorem~\ref{thm:small-k-c}.
		\Else
			\State Let $M=M[X]\uplus M[Y]\uplus M[Z]$ where $|X|=|Y|=2$ and $|Z|=k-4$.
			\If{there are no interactions between $X$ and $Y$}
				\For{each embedding $f_Z$ for $Z$}
					\State Independently compute the best embeddings $f_X$ for $X$ and $f_Y$ for $Y$.
				\EndFor
			\Else
				\State Relax the constraint $f_X(i)<f_Y(i+1)$ to $f_X(i)\neq f_Y(i+1)$.
				\For{each embedding $f_Z$ for $Z$}
					\State Compute the best pair $(f_X,f_Y)$ satisfying the relaxed constraints.
				\EndFor
			\EndIf
		\EndIf
	\EndFor
	\end{algorithmic}
\end{algorithm}

Now we are ready to describe the algorithm from Theorem~\ref{thm:at-most-7} (see also Pseudocode~\ref{alg:k7}).
For each feasible and irreducible connection $k$-pattern $M$, we compute the best embedding as follows.
If $M$ consists of at most two sequential swaps, we can use the algorithm in Theorem~\ref{thm:small-k-c}.
Otherwise, $M$ consists of three sequential swaps $M[X]$, $M[Y]$, $M[Z]$ such that $X\uplus Y\uplus Z = [k]$, $|X|=|Y|=2$ and $|Z|=k-4$.
For each embedding $f_X:X\rightarrow[n]$ of $X=\{i,j\}$ we create a 2-dimensional point $(f_X(i), f_X(j))$ with priority $\gain_{X}(f_X)$ and we put all the points in a range tree data structure $D_X$~\cite{PreparataS85}.
We build an analogous data structure for $Y$.
Next, for each embedding $f_Z$ for $Z$, we compute the best pair of embeddings $(f_X,f_Y)$ for $X$ and $Y$ as follows.

If there are no interactions between $X$ and $Y$,
we can find the best pair in $\Oh(\polylog n)$ time by independently picking the best embeddings for $X$ and $Y$ by querying the range trees $D_X$ and $D_Y$.
Indeed, first note that there is no index $i\in [k-1]$ such that $X=\{i,i+1\}$ because in such a case, both the 2-swap and the remaining $(k-2)$-swap have to be feasible (similarly for $Y$).
Since there are no interactions between $X$ and $Y$, we must have $i-1\in Z\cup\{0\}$ and $i+1\in Z\cup \{k+1\}$ for every $i\in X\cup Y$.
To find the best embedding $f_X$ of $X=\{i,j\}$, we query $D_X$ with the constraints $f_Z(i-1)<f_X(i)<f_Z(i+1)$ and $f_Z(j-1)<f_X(j)<f_Z(j+1)$, where we define $f_Z(0):=0$ and $f_Z(k+1):=n+1$.
We proceed analogously for $Y$.

Finally, assume there are interactions between $X$ and $Y$, so from Lemma~\ref{lem:interactions}, there is exactly one interaction.
W.l.o.g.~$i\in X$ and $i+1\in Y$.
Note that $i-1\in Z\cup\{0\}$ and $i+2\in Z\cup\{k+1\}$.
We first relax the constraint $f_Z(i-1)<f_X(i)<f_Y(i+1)<f_Z(i+2)$, where we define $f_Z(0):=0$ and $f_Z(k+1):=n+1$, to
three constraints $f_Z(i-1)<f_X(i)<f_Z(i+2)$, $f_Z(i-1)<f_Y(i+1)<f_Z(i+2)$, and $f_X(i)\neq f_Y(i+1)$.
We then drop the disturbing inequality constraint $f_X(i)\neq f_Y(i+1)$ by color-coding\footnote{Instead of color-coding,
we can adapt the range tree to support orthogonal range maximum queries with an additional constraint of the form $x\neq i$ by keeping one additional point in each node.
With this approach, we can avoid the additional $\log^2 n$ factor.
Because this paper does not focus on optimizing the $\polylog n$ factor, we do not touch on the details.}.
We color each vertex in $[n]$ in red or blue, and we independently pick the best embedding for $X$ (resp. $Y$) that uses only red (resp. blue) vertices.
By using a family of perfect hash functions~\cite{FredmanKS84}, we can construct a set of $\Oh(\log^2 n)$ colorings such that, for every pair of embeddings $f_X$ and $f_Y$,
there is at least one coloring that colors all the vertices in $f_X$ red and all the vertices in $f_Y$ blue.

We now obtain the best pair of embeddings $(f_X,f_Y)$ satisfying the relaxed constraints.
If the obtained $k$-swap is not improving, we immediately know that there are no improving $k$-moves that fit into $M$.
If it is improving and satisfies the original constraint, we are done.
Finally, if it is improving but does not satisfy the original constraint, it fits into the connection pattern $M'$ that is obtained from $M$ by swapping $i$ and $i+1$.
By Lemma~\ref{lem:relax1}, $M'$ is either feasible or reducible.
Because no improving $k$-swaps fit into reducible connection patterns, $M'$ has to be feasible.
We therefore obtain a $k$-move that is as good as the best $k$-move that fits into $M$.
This completes the proof of Theorem~\ref{thm:at-most-7}.

We next consider the case of $k=8$.
Note that, because Lemma~\ref{lem:interactions} and~\ref{lem:relax1} do not assume $k\leq 7$, the above algorithm can also compute the best improving $k$-move that can be decomposed into three sequential swaps of size $(2,2,k-4)$ for any fixed $k$ under the same assumption.
Moreover, any connection patterns of $8$-moves consisting of four 2-swaps are reducible because it always induces a pair of two 2-swaps that interact at least twice.
The remaining case for $k=8$ is only when the $8$-move can be decomposed into three sequential swaps of size $(2,3,3)$.
In order to tackle this case, we exploit the bounded-weight assumption as follows.
For each connection pattern $M=M[X]\uplus M[Y]\uplus M[Z]$ with $|X|=2$ and $|Y|=|Z|=3$, and for each embedding $f_Z$ for $Z$,
we want to compute the best pair of embeddings $f_X$ for $X$ and $f_Y$ for $Y$.
When all the weights are integers from $[W]:=\{1,\cdots,W\}$, the gain of $(f_X,M[X])$ is an integer from $[-2W,2W]:=\{-2W,-2W+1,\cdots,2W\}$,
and the gain $(f_Y,M[Y])$ is an integer from $[-3W,3W]$.
We therefore have only $\Oh(W^2)$ pairs of gains.
By guessing the pair of gains, the query of finding the \emph{best} pair can be reduced to the query of finding an \emph{arbitrary} pair, 
and the latter query can be efficiently answered by adapting the range tree.

\begin{algorithm}[t!]
	\caption{Quasi-linear-time algorithm for $k=8$}
	\label{alg:k8}
	\begin{algorithmic}[1]
	\For{each irreducible connection pattern $M$ of 8-moves}
		\If{$M$ consists of three sequential swaps of size $(2,3,3)$}
			\State Let $M=M[X]\uplus M[Y]\uplus M[Z]$ with $X=\{i,j\}$.
			\If{$\{i-1,i+1\}\subseteq Y$ and $\{j-1,j+1\}\subseteq Z$}
				\State Relax $f_Y(i-1)<f_X(i)<f_Y(i+1)$ to $f_Y(i-1)\neq f_X(i) \neq f_Y(i+1)$.
				\For{each embedding $f_Z$ for $Z$}
					\State Compute the best pair $(f_X,f_Y)$ satisfying the relaxed constraints.
				\EndFor
			\Else
				\State We can assume $|\{i-1,i+1\}\cap Y|\leq 1$ and $|\{j-1,j+1\}\cap Y|\leq 1$.
				\For{each embedding $f_Z$ for $Z$}
					\For{$g_X\in [-2W,2W]$}
						\For{$g_Y\in [-3W,3W]$}
							\State Find a pair of embeddings $(f_X,f_Y)$ of gains $(g_X,g_Y)$.
						\EndFor
					\EndFor
				\EndFor
			\EndIf
		\Else
			\State Apply the algorithm in Theorem~\ref{thm:at-most-7}.
		\EndIf
	\EndFor
	\end{algorithmic}
\end{algorithm}

We now describe details of the algorithm for $k=8$ to prove Theorem~\ref{thm:8opt} (see also Pseudocode~\ref{alg:k8}).
For each irreducible connection pattern $M$ of $8$-moves, we compute the best embedding as follows.
If the size of the sequential moves in $M$ is not $(2,3,3)$, as we have discussed above, we can use the algorithm in Theorem~\ref{thm:at-most-7}.
Let $[k]=X\uplus Y\uplus Z$ with $|X|=2$ and $|Y|=|Z|=3$, and let $X=\{i,j\}$.
We take $Y$ and $Z$ so that the number of interactions between $X$ and $Y$ is not greater than the number of interactions between $X$ and $Z$.
Because the total number of interactions of $X$ is at most four, the number of interactions between $X$ and $Y$ is at most two.
We consider two cases.

(Case 1) If $\{i-1,i+1\}\subseteq Y$ and $\{j-1,j+1\}\subseteq Z$ hold, we will show that we can relax the constraint $f_Y(i-1)<f_X(i)<f_Y(i+1)$ to $f_Y(i-1)\neq f_X(i) \neq f_Y(i+1)$ as we did in Theorem~\ref{thm:at-most-7}.
We use the following lemma that corresponds to Lemma~\ref{lem:relax1} used in Theorem~\ref{thm:at-most-7}.
\begin{lemma}\label{lem:relax2}
Let $M$ be an irreducible connection pattern of $8$-moves that consists of three sequential moves $[k]=X\uplus Y\uplus Z$ with $X=\{i,j\}$ and $|Y|=|Z|=3$ and satisfies
$\{i-1,i+1\}\subseteq Y$ and $\{j-1,j+1\}\subseteq Z$.
Then for at least one of the pair $(a,A)\in\{(i,Y), (j,Z)\}$, all the following conditions hold.
\begin{enumerate}
	\item $a-2,a+2\not\in A$.
	\item The connection pattern obtained from $M$ by swapping $a$ and $a+1$ is either feasible or reducible.
	\item The connection pattern obtained from $M$ by swapping $a$ and $a-1$ is either feasible or reducible.
\end{enumerate}
\end{lemma}
\begin{proof}
Because the number of connection $8$-patterns is finite, we can prove the lemma by checking the statement against all the connection patterns.
In Appendix~\ref{sec:app-code}, we attach a source code of a Python program that generates all the connection pattern and checks the statement against each of them.
The same source code is also available at \url{https://github.com/wata-orz/TSP_analysis}.
There are 645120 connection patterns of 8-moves, and among them, 136 patterns meet the precondition in the lemma.
We checked that all of them satisfy the statement in the lemma.
\end{proof}

Assume that $(i,Y)$ satisfies all the three conditions in the lemma.
Then we have $i-2\in Z\cup \{0\}$ and $i+2\in Z\cup \{k+1\}$.
For each embedding $f_Z$ for $Z$, we compute the best pair of embeddings $(f_X,f_Y)$ as follows.
We first relax the constraint $f_Z(i-2)<f_Y(i-1)<f_X(i)<f_Y(i+1)<f_Z(i+2)$, where we define $f_Z(0):=0$ and $f_Z(k+1):=n+1$, to
four constraints $f_Z(i-2)<f_X(i)<f_Z(i+2)$, $f_Z(i-2)<f_Y(i-1)$, $f_Y(i+1)<f_Z(i+2)$, and $f_Y(i-1)\neq f_X(i)\neq f_Y(i+1)$.
We then drop the disturbing inequality constraint $f_Y(i-1)\neq f_X(i)\neq f_Y(i+1)$ by color-coding.
We now obtain the best pair of embeddings $(f_X,f_Y)$ satisfying the relaxed constraints.
If the obtained $k$-swap is not improving, we immediately know that there are no improving $k$-moves fit into $M$.
If it is improving but does not satisfy the original constraint, it fits into the connection pattern $M'$ that is obtained from $M$ by swapping either $i$ and $i+1$ or $i$ and $i-1$.
From Lemma~\ref{lem:relax2}, $M'$ is either feasible or reducible.
Because no improving $k$-swaps fit into reducible connection patterns, $M'$ has to be feasible.
We therefore obtain a $k$-move that is as good as the best $k$-move that fits into $M$.

(Case 2) If the condition of the case 1 is not met, we have $|\{i-1,i+1\}\cap Y|\leq 1$ and $|\{j-1,j+1\}\cap Y|\leq 1$.
For each embedding $f_Z$ for $Z$, we compute the best pair $(f_X,f_Y)$ as follows.
We first reduce the problem of finding the \emph{best} pair to the problem of finding an \emph{arbitrary} pair by guessing the gains.
When all the weights are integers from $[W]$, the gain $g_X$ of $f_X$ is an integer from $[-2W,2W]$, and the gain $g_Y$ of $f_Y$ is an integer from $[-3W,3W]$.
We therefore have only $\Oh(W^2)$ choices.
For each embedding $f_X$ with the guessed gain $g_X$, we create the corresponding 2-dimensional point,
and for each embedding $f_Y$ with the guessed gain $g_Y$, we create the corresponding 3-dimensional point.
Because the number of interactions between $X$ and $Y$ is at most two, and there is no index $i\in X$ such that $\{i-1,i+1\}\subseteq Y$, we can answer the existence of the desired pair of points by using the following data structure.
This completes the proof of Theorem~\ref{thm:8opt}.

\begin{lemma}
Given a set $P$ of three dimensional points and a set $Q$ of two dimensional points,
we can construct, in $\Oh(n \polylog n)$ time for $n=|P|+|Q|$, a data structure to answer each of the following queries in $\Oh(\polylog n)$ time.
\begin{enumerate}
	\item Given four ranges $(R_x, R_y^p, R_y^q, R_z)$, find $p=(p_x,p_y,p_z)\in P$ and $q=(q_x,q_y)\in Q$ satisfying $p\in R_x\times R_y^p\times R_z$, $q\in R_x\times R_y^q$, and $p_x<q_x$.
	\item Given three ranges $(R_x, R_y, R_z)$, find $p=(p_x,p_y,p_z)\in P$ and $q=(q_x,q_y)\in Q$ satisfying $p\in R_x\times R_y\times R_z$, $q\in R_x\times R_y$, $p_x<q_x$, and $p_y<q_y$.
\end{enumerate}
\end{lemma}
\begin{proof}
The first query is easy.
We construct the orthogonal range minimum data structure for $P$ with priority $p_x$ and the orthogonal range maximum data structure for $Q$ with priority $q_x$.
Given the query, we compute the pair of points $(p,q)$ with the minimum $p_x$ and the maximum $q_x$, and if $p_x<p_x$ holds, we return it.

We then focus on the second query.
We construct two range trees $T_P$ and $T_Q$ for $P$ and $Q$, respectively, so that the first two layers of $T_P$ has the same structure as $T_Q$, i.e.,
$T_P$ has a node (in the first layer) with a range $[L_x,U_x)$ if and only if $T_Q$ has a node with the same range $[L_x,U_x)$,
and $T_P$ has a node (in the second layer) with a range $[L_x,U_x)\times[L_y,U_y)$ if and only if $T_Q$ has a node with the same range $[L_x,U_x)\times[L_y,U_y)$.
For each node in $T_P$ (resp. $T_Q$), we keep two points in the associated range with the minimum (resp. maximum) $x$- or $y$-coordinate.
By the standard construction of the range tree, the total number of nodes is $\Oh(n \log^2 n)$, and the construction time is $\Oh(n \log^3 n)$.
We then, for each node in the third layer of $T_P$, compute a pair of points $p\in P$ and $q\in Q$ both of which are contained in the associated range and satisfy $p_x<q_x$ and $p_y<q_y$ bottom-up as follows.
For leaf node, because there is only one point $p\in P$ in the associated range, we can find such a point $q$ in $\Oh(\log^2 n)$ time by the orthogonal range existance query $(p_x,\infty]\times(p_y,\infty]$ to $T_Q$.
For non-leaf node, we pick a pair contained in children if exists.
This finishes the construction of the data structure.

We then describe how to answer the query.
Starting from the root, we can recursively compute two points in $P$ (resp. $Q$) with the minimum (resp. maximum) $x$- or $y$-coordinate that is contained in the intersection of the given range and the current range
by the standard query processing of the range tree.
We additionally do the following to find the desired pair.

When processing both the children of a first-layer node, we additionally do the following post-process.
We can find the desired pair $(p,q)$ such that both of $p$ and $q$ are contained in the range of one of the children by recursion.
We therefore only need to find the desired pair $(p,q)$ such that one is contained in the range of one child, and the other is contained in the range of the other child.
Because $p_x<q_x$, $p$ has to be in the left child and $q$ has to be in the right child.
Let $p$ (resp. $q$) be the point with the minimum $p_y$ (resp. maximum $q_y$) obtained from the left (resp. right) child.
Because $p_x<q_x$ already holds, there exists the desired pair if and only if $p_y<q_y$ holds.

When processing both the children of a second-layer node, we similarly do the following post-process.
We want to find the desired pair $(p,q)$ such that $p$ is in the left child and $q$ is in the right child.
Let $p$ (resp. $q$) be the point with the minimum $p_x$ (resp. maximum $q_x$) obtained from the left (resp. right) child.
Because $p_y<q_y$ already holds, there exists the desired pair if and only if $p_x<q_x$ holds.

When processing a third-layer node whose associated range is contained in the given range,
we need to find the desired pair $(p,q)$ both of which are contained in the associated range.
This can be done by checking the precomputed pair.

The number of nodes processed by each query is $\Oh(\log^3 n)$, and we can process each node in constant time.
Therefore the query time is $\Oh(\log^3 n)$.
\end{proof}

\section{Lower bound for $k=9$} \label{sec:small-k-lowerbound}
The starting point for our reduction is the following problem (see Fig.~\ref{fig:instance} for an exemplary instance).

\begin{figure}
\centering      
\includegraphics[scale=1]{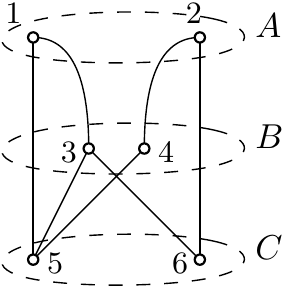}
\caption{\label{fig:instance}An instance of {\sc Triangle Detection}}
\end{figure}

\defparproblem{\probTriangle}
{An undirected graph~$H$ whose vertex set $V(H)$ is partitioned into $A \cup B \cup C$.}
{$m := |E(H)|$.}
{Is there a triple $(a,b,c) \in A \times B \times C$ such that $\{ab,ac,bc\} \subseteq E(H)$?}

We assume without loss of generality that $A$, $B$, and $C$ are three independent sets, so that finding such a triple is equivalent to finding a triangle in the graph $H$. By simple reductions that incur only a constant blow-up in the number of vertices and edges, this problem is equivalent to determining whether a graph has a triangle or not.

\begin{assumption}[Triangle hypothesis \cite{AbboudW14}]
% There is a constant $\delta > 1$ such that \probTriangle is in time $\Omega(m^\delta)$.
There is a fixed $\delta > 0$ such that, in the Word RAM model with words of~$\Oh(\log n)$ bits, any algorithm requires~$m^{1+\delta - o(1)}$ time in expectation to detect whether an~$m$-edge graph contains a triangle.
\end{assumption}

It should be noted that one can solve \probTriangle in time $\Oh(n^\omega)$ where $n$ is the number of vertices and $\omega \leq 2.373$ is the best-known exponent for matrix multiplication.
Alon et al. \cite{AlonYZ97} found an elegant win-win argument to solve \probTriangle in time $\Oh(m^{\frac{2\omega}{\omega+1}})$: the 3-vertex paths in which the middle vertex has degree less than $m^{\frac{\omega-1}{\omega+1}}$ can be listed in time $\Oh(m \cdot m^{\frac{\omega-1}{\omega+1}})=\Oh(m^{\frac{2\omega}{\omega+1}})$, and for each, one can check if they form a triangle, whereas the number of vertices of degree greater than $m^{\frac{\omega-1}{\omega+1}}$ is at most $m^{\frac{2}{\omega+1}}$, so one can detect a triangle in time $\Oh(m^{\frac{2\omega}{\omega+1}})$ in the subgraph that they induce.
After more than two decades, this is still the best worst-case running time (when $n^\omega =\Omega(m^{\frac{2\omega}{\omega+1}})$).
This suggests that the triangle hypothesis is likely to hold.
Moreover, if one thinks that the above scheme yields the best possible running time and that $\omega$ will eventually reach $2$, then exponent $4/3$ could be the \emph{right answer} for \probTriangle parameterized by the number of edges. The following is implied by~\cite[Conjecture 2]{AbboudW14} (since~$\omega \geq 2$), in the regime $m = \Theta(n^{3/2})$ (so that $\Oh(n^2)$ and $\Oh(m^{4/3})$ coincide).

\begin{assumption}
 %\probTriangle is in time $\Omega(m^{4/3})$.
 %The statement called 'strong triangle' in the FOCS paper is: 
In the Word RAM model with words of~$\Oh(\log n)$ bits, any algorithm requires $m^{4/3 - o(1)}$ time in expectation to detect whether an $m$-edge $\Theta(m^{2/3})$-node graph contains a triangle.
\end{assumption}

%\begin{assumption}[Strong triangle hypothesis \cite{AbboudW14}]
 %\probTriangle is in time $\Omega(m^{4/3})$.
 %The statement called 'strong triangle' in the FOCS paper is: 
%In the Word RAM model with words of~$\Oh(\log n)$ bits, any algorithm requires $ \min \{ n^{2 - o(1)}, m^{4/3 - o(1)} \}$ time in expectation to detect whether an $n$-node $m$-edge graph contains a triangle.
%\end{assumption}

We show that \textsc{Subcubic} \probOPTDec{9} parameterized by the number of vertices is as hard as \probTriangle parameterized by the number of edges, by providing a linear-time reduction from the latter to the former.
In light of \cref{thm:at-most-7}, this implies that \textsc{Bounded-Degree} \probOPTDec{8} is the only remaining open case where a quasi-linear algorithm is not known but also not ruled out by a standard fine-grained complexity assumption.

\begin{lemma}
There is an $\Oh(m)$-time reduction from \probTriangle on $m$-edge graphs to \textsc{Subcubic} \probOPTDec{9} on $\Oh(m)$-vertex undirected graphs with edge weights in $\{1,2\}$.
\end{lemma}

\begin{proof}
  From a tripartitioned instance of \probTriangle $H=(A \cup B \cup C, E(H))$ with $m$ edges, we build a subcubic graph $G$ with $\Theta(m)$ vertices, an edge-weight function $w: E(G) \rightarrow \{1,2\}$, and a Hamiltonian cycle $\mathcal C$.
  From $\mathcal C$, there is a swap of up to 9 edges (i.e., up to 9 deletions and the same number of additions) which results in a lighter Hamiltonian cycle if and only if $H$ has a triangle.

  \paragraph*{Overall construction of $G$.}
  We will build $G$ by adding chords to the cycle $\mathcal C$.
  Henceforth, a \emph{chord} is an edge of $G$ which is not in $\mathcal C$.
  It is helpful to think of $\mathcal C$ as a (subdivided) triangle whose three sides correspond to $A$, $B$, and $C$, which we call the $A$-side (left), $B$-side (right), and $C$-side (bottom), respectively.
  We will only name the edges of $G$ (and not the vertices), since the problem is more efficiently described in terms of edges.
  We will define some sequential 3-swaps (we recall that a sequential $i$-swap is a closed walk of length $2i$ alternating edges of $E(\mathcal C)$ and edges of $E(G) \setminus E(\mathcal C)$). 
  Eventually, all the edges that are not in a described sequential 3-swap are incident to a vertex of degree 2, making them undeletable.
  (One can also enforce that by subdividing every irrelevant edge once.)
  
  The improving 9-move, should there be a triangle $abc$ in $H$, will consist of a sequence of three 3-swaps.
  More precisely, it consists of one improving 3-swap, which splits $\mathcal C$ into three cycles respectively containing:
  \begin{enumerate}[(1)]
  \item a part of the vertex gadget of some $a \in A$,
  \item the part of the $B$-side below the vertex gadget of $b$, as well as the $C$-side, and
  \item the part of the $B$-side above the vertex gadget of some $b \in N_H(a) \cap B$.
  \end{enumerate}
  This decreases the total weight by 1.
  Then a neutral 3-swap reconnects (1) and (2) together, but also detaches (4) a part of the vertex gadget of some $c \in N_H(a) \cap C$.
  Finally a neutral 3-swap glues (3), (1)$+$(2), and (4) together, provided $bc \in E(H)$.
  This results in a new Hamiltonian cycle of length $w(\mathcal C)-1$.

  There will be relatively few edges of weight 2.
  To simplify the presentation, every edge is of weight 1, unless specified otherwise.
  Let $\vec{H}$ be the directed graph obtained from $H$ by orienting its edges from $A$ to $B$, from $B$ to $C$, and from $C$ to $A$.
  Note that finding a directed triangle in $\vec{H}$ is equivalent to finding a triangle in $H$.   

  \paragraph*{Vertex scopes, extended scopes, and nested chords.}
  For $(X, Y) \in \{(A, B), (B, C), (C, A)\}$, we set $Z := \{A, B, C\} \setminus \{X, Y\}$ and we do the following as a preparatory step to encode the arcs of $\vec{H}$.
  Each vertex $v \in X$ is given a (pairwise vertex-disjoint) subpath $I_v$ of $\mathcal C$, called the \emph{extended scope} of $v$, with $|I_v| := 6(|N_H(v) \cap Y|) + 3(|N_H(v) \cap Z|) - 1$ vertices.
  %\todo[inline]{Bart writes: are these numbers correct? They imply that a vertex with 1 neighbor in Y and 1 in Z gets a segment of only 3 vertices, but then an edge from the second to penultimate would not be a chord but 'parallel' to an edge of the cycle. Does the segment have to be larger, or do we want to rule out low-degree vertices because you can efficiently find triangles through them anyway?}
  We think of $I_v$ as being displayed from left to right with the leftmost vertex of index $1$, and the rightmost one of index $|I_v|$. 
  The extended scopes of the vertices of $A$, $B$, and $C$ occupy respectively the $A$-side, $B$-side, and $C$-side.
  In what follows, it will be more convenient to have a \emph{circular} notion of \emph{left} and \emph{right}.
  Starting from the bottom corner of the $A$-side, and going clockwise to the top corner of the $A$-side, then down to the bottom corner of the $B$-side, the relative left and right within the $A$-side and the $B$-side coincide with the usual notion as displayed in Figure~\ref{fig:overallTriangle}.
  But then closing the loop from the right corner of the $C$-side to its left corner, left and right are switched: the closer to the bottom corner of $A$ (resp.~$B$), the more ``right'' (resp.~``left''). 
  
  Each vertex $v \in X$ has $|N_H(v) \cap Y|$ nested chords spaced out every three vertices.
  More precisely, the second vertex of $I_v$ is adjacent to the penultimate, the fifth to the one of index $|I_v|-4$, the eighth to the one of index $|I_v|-7$, and so on, until $|N_H(v) \cap Y|$ chords are drawn.
  Each of these chords is associated to an edge $vy \in E(\{v\},Y)$, and is denoted by $\underline{v}y$.
  A vertex just to the right of the left endpoint, or just to the left of the right endpoint, of such a chord will remain of degree 2 in $G$.
  This is the case of the vertices of index $3, 6, \ldots$ and $|I_v| - 2, |I_v|-5, \ldots$ in $I_v$.
  We call $l^-(v,y)$ (resp.~$r^-(v,y)$) the edge of $I_v$ incident to both the left endpoint of $\underline{v}y$ and the vertex just to its left (resp.~right endpoint of $\underline{v}y$ and the vertex just to its right).
  Both endpoints of $l^-(v,y)$ and of $r^-(v,y)$ will eventually have degree 3 in $G$.
  
  The chord linking the most distant vertices in $I_v$ is called the \emph{outermost} chord, while the one linking the closest pair is called the \emph{innermost} chord.
  We also say that a chord $e$ is \emph{wider} than a chord $e'$ if $e$ links a farther pair on $I_v$ than $e'$ does. 
  The central path $J_v \subset I_v$ on $|I_v| - (6|N_H(v) \cap Y| - 4) = 3(|N_H(v) \cap Z|+1)$ vertices,
  %\todo[inline]{With the current definition for~$|I_v|$ this is negative if a vertex has more neighbors in $Y$ than in $Z$}
  surrounded by the innermost chord, is called the \emph{scope} of $v$.
  We map in one-to-one correspondence the edges of $E(\{v\},Z)$ to every three edges of $J_v$ starting from the third edge (that is, the third, sixth, and so on).
  Note that we have the exact space to do so, since $|J_v|=3(|N_H(v) \cap Z|+1)$.
  We denote by $z\underline{v}$ the edge in $J_v$ corresponding to the edge $vz \in E(\{v\},Z)$.

  \paragraph*{Encoding the arcs of $\vec{H}$.}
  The last step to encode the arcs of $\vec{H}$, or equivalently the edges of $H$, is the following.
  Keeping the notations of the previous paragraphs, for every edge $xy \in E(X,Y)$, we add two chords (of weight 1): one chord $l^+(x,y)$ between the left endpoint of $l^-(x,y)$ and the right endpoint of $x\underline{y}$ and one chord $r^+(x,y)$ between the right endpoint of $r^-(x,y)$ and the left endpoint of $x\underline{y}$.
  We finish the construction of $G$ (and $\mathcal C$) by subdividing each edge between consecutive extended scopes once, to make the resulting edges undeletable.
  The edges $l^-(a,b)$ for $(a,b) \in A \times B$ get weight 2, while all the other edges of $E(G)$ get weight 1.
  This finishes the construction of $(G,w,\mathcal C)$.
  See Figure~\ref{fig:overallTriangle} for an illustration.

  \paragraph*{Improving and neutral 3-swaps.}
  For each $(x,y) \in E(\vec{H})$, we denote by $S(x,y)$ the 3-swap $(\{x\underline{y}, l^-(x,y), r^-(x,y)\},$ $\{\underline{x}y, l^+(x,y),$ $ r^+(x,y)\})$.
  For $(X,Y) \in \{(A,B), (B,C),$ $(C,A)\}$, we define the set of 3-swaps $S(X,Y) := \underset{xy \in E(X,Y)}{\bigcup} S(x,y)$, and $\mathcal S := S(A,B) \cup S(B,C) \cup S(C,A)$.

  Note that all the 3-swaps of $S(A,B)$ are improving.
  They gain 1 since $l^-(a,b)$ has weight 2 for any $(a,b) \in A \times B$.
  On the other hand, all the 3-swaps of $S(B,C)$ and $S(C,A)$ are neutral.
  The edges added in swaps of $\mathcal S$ partition the chords of $G$, and the open neighborhood of the six vertices involved in every swap are six vertices of degree 2 in $G$.
  Therefore, all the possible swaps are in the set $\mathcal S$, they are on vertex-disjoint sets of vertices, and any move is a sequence of $3$-swaps of $\mathcal S$.

  The vertices of $\mathcal C$ are incident to at most one chord.
  Hence the graph $G$ is subcubic.
  It has $\sum_{v \in V(H)}1+|I_v| \leqslant 9|E(H)| + |V(H)| = \Theta(m)$ vertices and $(G,w,\mathcal C)$ takes $\Theta(m)$-time to build.
  To summarize, we defined a linear reduction from \textsc{Triangle Detection} with parameter~$m$ to \textsc{Subcubic} \probOPTDec{9} with parameter~$n$.
  So a quasi-linear algorithm for the latter would yield an unlikely quasi-linear algorithm for the former.
  We now check that the reduction is correct.

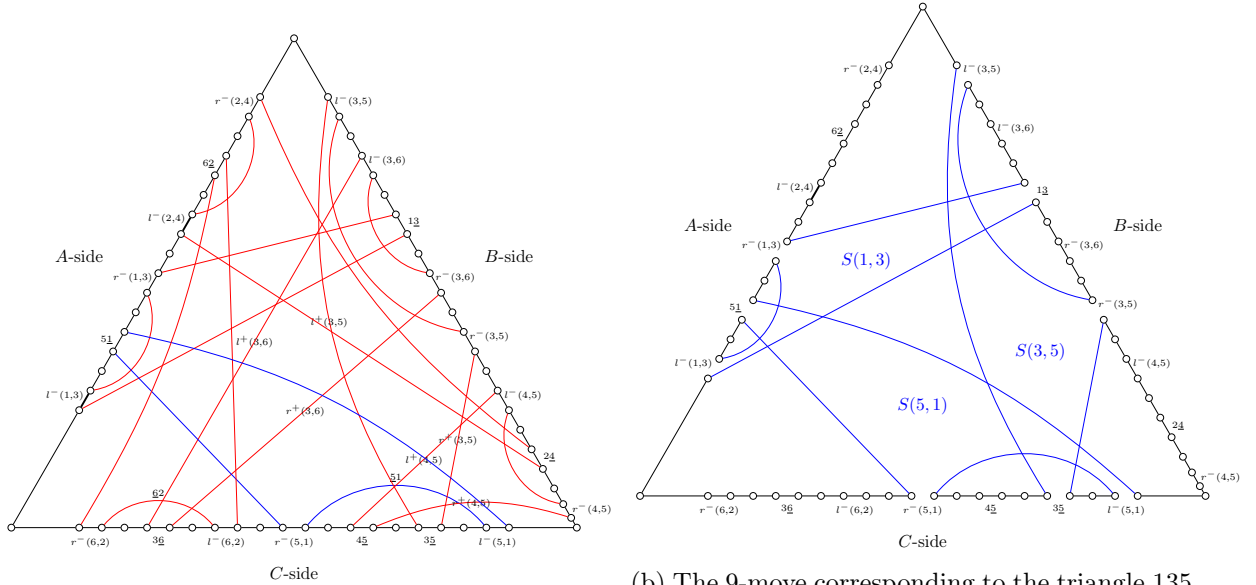
\begin{figure}
\centering      
\begin{minipage}{0.45\textwidth}
\centering
\resizebox{235pt}{!}{
\begin{tikzpicture}
%C-side
\foreach \l/\x in {CB/0, rr51/1.5,rl51/2,d53/2.5,r53/3,l53/3.5,d54/4,r54/4.5,l54/5,e54/5.5,lr51/6,ll51/6.5, s56/7, rr62/7.5,rl62/8,d63/8.5,r63/9,l63/9.5,e63/10,lr62/10.5,ll62/11, AC/12.5}{
  \node[vp] (\l) at (-\x,0) {} ;
}
\draw (CB) -- (rr51) -- (rl51) -- (d53) -- (r53) -- (l53) -- (d54) -- (r54) -- (l54) -- (e54) -- (lr51) -- (ll51) -- (s56) -- (rr62) -- (rl62) -- (d63) -- (r63) -- (l63) -- (e63) -- (lr62) -- (ll62) -- (AC) ;
\foreach \t/\x in {l^-(5\text{,}1)/1.75,3\underline{5}/3.25,4\underline{5}/4.75,r^-(5\text{,}1)/6.25, l^-(6\text{,}2)/7.75, 3\underline{6}/9.25, r^-(6\text{,}2)/10.75}{
   \node at (-\x,-0.3) {\tiny{$\t$}} ;
}
\foreach \i/\j/\c in {rl51/lr51/\colob,rl62/lr62/\colo}{
  \path[color=\c,-] (\i) edge[bend left=-\bendedge] (\j) ;
}
\node at (-9.25,0.75) {\tiny{\underline{6}2}} ;
\node at (-4,1.1) {\tiny{\underline{5}1}} ;

%B-side
\begin{scope}[rotate=120]
\foreach \l/\x in {rr45/0.25,rl45/0.6,d42/1,r42/1.5,l42/2,e42/2.5,lr45/3,ll45/3.5,s34/4,rr35/4.5,rl35/5,d356/5.5,rr36/6,rl36/6.5, d31/7, r31/7.5,l31/8,e31/8.5,lr36/9,ll36/9.5,d365/10,lr35/10.5,ll35/11, AB/12.5}{
  \node[vp] (\l) at (\x,0) {} ;
}
\draw (CB) -- (rr45) -- (rl45) -- (d42) -- (r42) -- (l42) -- (e42) -- (lr45) -- (ll45) -- (s34) -- (rr35) -- (rl35) -- (d356) -- (rr36) -- (rl36) -- (d31) -- (r31) -- (l31) -- (e31) -- (lr36) -- (ll36) -- (d365) -- (lr35) -- (ll35) -- (AB) ;
\foreach \t/\x in {l^-(4\text{,}5)/3.25,r^-(3\text{,}5)/4.75,r^-(3\text{,}6)/6.25, l^-(3\text{,}6)/9.25, l^-(3\text{,}5)/10.75}{
   \node at (\x - 0.1,-0.45) {\tiny{$\t$}} ;
}
\node at (0.2,-0.5) {\tiny{$r^-(4\text{,}5)$}} ;       
\foreach \t/\x in {2\underline{4}/1.75, 1\underline{3}/7.75}{
   \node at (\x - 0.1,-0.25) {\tiny{$\t$}} ;
}
\foreach \i/\j in {rl35/lr35,rl36/lr36,rl45/lr45}{
  \path[color=\colo,-] (\i) edge[bend left=\bendedge] (\j) ;
}
\end{scope}

%A-side
\begin{scope}[xshift=-12.5cm,rotate=60]
\foreach \l/\x in {ll13/3,lr13/3.5,d15/4,l15/4.5,r15/5,e15/5.5,rl13/6,rr13/6.5, s12/7, ll24/7.5,lr24/8,d26/8.5,l26/9,r26/9.5,e26/10,rl24/10.5,rr24/11}{
  \node[vp] (\l) at (\x,0) {} ;
}
\draw (AC) -- (ll13) ;
\draw[very thick] (ll13) -- (lr13) ;
\draw (lr13) -- (d15) -- (l15) -- (r15) -- (e15) -- (rl13) -- (rr13) -- (s12) -- (ll24) ;
\draw[very thick] (ll24) -- (lr24) ;
\draw (lr24) -- (d26) -- (l26) -- (r26) -- (e26) -- (rl24) -- (rr24) -- (AB) ;
\foreach \t/\x in {l^-(1\text{,}3)/3.25,r^-(1\text{,}3)/6.25, l^-(2\text{,}4)/7.75, r^-(2\text{,}4)/10.75}{
   \node at (\x - 0.1,0.45) {\tiny{$\t$}} ;
}
\foreach \t/\x in {5\underline{1}/4.75, 6\underline{2}/9.25}{
   \node at (\x - 0.1,0.25) {\tiny{$\t$}} ;
}
\foreach \i/\j in {lr13/rl13,lr24/rl24}{
  \path[color=\colo,-] (\i) edge[bend left=-\bendedge] (\j) ;
}

%edges
\foreach \x/\y/\l in {rr35/r53/r^+(3\text{,}5), ll36/l63/l^+(3\text{,}6),rr36/r63/r^+(3\text{,}6), ll45/l54/l^+(4\text{,}5)}{        
  \draw[color=\colo] (\x) edge node[black] {\tiny{$\l$}} (\y) ;
}
\foreach \x/\y/\l in {ll35/l53/l^+(3\text{,}5), rr45/r54/r^+(4\text{,}5)}{        
  \draw[color=\colo] (\x) edge[bend right=20] node[black] {\tiny{$\l$}} (\y) ;
}
\foreach \x/\y/\c in {ll13/r31/\colo,rr13/l31/\colo, ll24/r42/\colo, ll51/l15/\colob, rr62/r26/\colo}{        
  \draw[color=\c] (\x) -- (\y) ;
}
\foreach \x/\y/\c in {rr24/l42/\colo, rr51/r15/\colob}{        
  \draw[color=\c] (\x) edge[bend right=15] (\y) ;
}
\foreach \x/\y in {ll62/l26}{        
  \draw[color=\colo] (\x) edge[bend right=10] (\y) ;
}
\end{scope}

%sides
\node at (-6.25,-1) {$C$-side} ;
\node at (-1.5,6) {$B$-side} ;
\node at (-11,6) {$A$-side} ;
\end{tikzpicture}
}
\subcaption{The construction for the instance of Figure~\ref{fig:instance}. Edges of $\mathcal C$ are in black, chords are in red, bold edges are the ones with weight 2. The three chords in blue are the edges to add to perform the neutral 3-swap $S(5,1)$ of $S(C,A)$.
%To avoid cluttering the figure, we only labeled the chords on the $C$-side and in the interaction between the $B$-side and the $C$-side.
}
\label{fig:overallTriangle}
\end{minipage}
\qquad
\begin{minipage}{0.45\textwidth}
\centering
\resizebox{235pt}{!}{
\begin{tikzpicture}
%C-side
\foreach \l/\x in {CB/0, rr51/1.5,rl51/2,d53/2.5,r53/3,l53/3.5,d54/4,r54/4.5,l54/5,e54/5.5,lr51/6,ll51/6.5, s56/7, rr62/7.5,rl62/8,d63/8.5,r63/9,l63/9.5,e63/10,lr62/10.5,ll62/11, AC/12.5}{
  \node[vp] (\l) at (-\x,0) {} ;
}
\draw (CB) -- (rr51) ;
\draw (rl51) -- (d53) -- (r53) ;
\draw (l53) -- (d54) -- (r54) -- (l54) -- (e54) -- (lr51) ;
\draw (ll51) -- (s56) -- (rr62) -- (rl62) -- (d63) -- (r63) -- (l63) -- (e63) -- (lr62) -- (ll62) -- (AC) ;
\foreach \t/\x in {l^-(5\text{,}1)/1.75,3\underline{5}/3.25,4\underline{5}/4.75,r^-(5\text{,}1)/6.25, l^-(6\text{,}2)/7.75, 3\underline{6}/9.25, r^-(6\text{,}2)/10.75}{
   \node at (-\x,-0.3) {\tiny{$\t$}} ;
}
\foreach \i/\j/\c in {rl51/lr51/\colob}{
  \path[color=\c,-] (\i) edge[bend left=-\bendedge] (\j) ;
}

%B-side
\begin{scope}[rotate=120]
\foreach \l/\x in {rr45/0.25,rl45/0.6,d42/1,r42/1.5,l42/2,e42/2.5,lr45/3,ll45/3.5,s34/4,rr35/4.5,rl35/5,d356/5.5,rr36/6,rl36/6.5, d31/7, r31/7.5,l31/8,e31/8.5,lr36/9,ll36/9.5,d365/10,lr35/10.5,ll35/11, AB/12.5}{
  \node[vp] (\l) at (\x,0) {} ;
}
\draw (CB) -- (rr45) -- (rl45) -- (d42) -- (r42) -- (l42) -- (e42) -- (lr45) -- (ll45) -- (s34) -- (rr35) ;
\draw (rl35) -- (d356) -- (rr36) -- (rl36) -- (d31) -- (r31) ;
\draw (l31) -- (e31) -- (lr36) -- (ll36) -- (d365) -- (lr35) ;
\draw (ll35) -- (AB) ;
\foreach \t/\x in {l^-(4\text{,}5)/3.25,r^-(3\text{,}5)/4.75,r^-(3\text{,}6)/6.25, l^-(3\text{,}6)/9.25, l^-(3\text{,}5)/10.75}{
   \node at (\x - 0.1,-0.45) {\tiny{$\t$}} ;
}
\node at (0.2,-0.5) {\tiny{$r^-(4\text{,}5)$}} ;       
\foreach \t/\x in {2\underline{4}/1.75, 1\underline{3}/7.75}{
   \node at (\x - 0.1,-0.25) {\tiny{$\t$}} ;
}
\foreach \i/\j in {rl35/lr35}{
  \path[color=\colob,-] (\i) edge[bend left=\bendedge] (\j) ;
}
\end{scope}

%A-side
\begin{scope}[xshift=-12.5cm,rotate=60]
\foreach \l/\x in {ll13/3,lr13/3.5,d15/4,l15/4.5,r15/5,e15/5.5,rl13/6,rr13/6.5, s12/7, ll24/7.5,lr24/8,d26/8.5,l26/9,r26/9.5,e26/10,rl24/10.5,rr24/11}{
  \node[vp] (\l) at (\x,0) {} ;
}
\draw (AC) -- (ll13) ;
\draw (lr13) -- (d15) -- (l15) ;
\draw (r15) -- (e15) -- (rl13) ;
\draw (rr13) -- (s12) -- (ll24) ;
\draw[very thick] (ll24) -- (lr24) ;
\draw (lr24) -- (d26) -- (l26) -- (r26) -- (e26) -- (rl24) -- (rr24) -- (AB) ;
\foreach \t/\x in {l^-(1\text{,}3)/3.25,r^-(1\text{,}3)/6.25, l^-(2\text{,}4)/7.75, r^-(2\text{,}4)/10.75}{
   \node at (\x - 0.1,0.45) {\tiny{$\t$}} ;
}
\foreach \t/\x in {5\underline{1}/4.75, 6\underline{2}/9.25}{
   \node at (\x - 0.1,0.25) {\tiny{$\t$}} ;
}
\foreach \i/\j in {lr13/rl13}{
  \path[color=\colob,-] (\i) edge[bend left=-\bendedge] (\j) ;
}

%edges
\foreach \x/\y in {rr35/r53}{        
  \draw[color=\colob] (\x) edge (\y) ;
}
\foreach \x/\y in {ll35/l53}{        
  \draw[color=\colob] (\x) edge[bend right=20] (\y) ;
}
\foreach \x/\y/\c in {ll13/r31/\colob,rr13/l31/\colob}{        
  \draw[color=\c] (\x) -- (\y) ;
}
\foreach \x/\y/\c in {ll51/l15/\colob}{        
  \draw[color=\c] (\x) -- (\y) ;
}
\foreach \x/\y/\c in {rr51/r15/\colob}{        
  \draw[color=\c] (\x) edge[bend right=15] (\y) ;
}
\end{scope}

%sides
\node at (-6.25,-1) {$C$-side} ;
\node at (-1.5,6) {$B$-side} ;
\node at (-11,6) {$A$-side} ;

%swaps
\node[blue] at (-7.5,5.2) {$S(1,3)$} ;
\node[blue] at (-6.25,2) {$S(5,1)$} ;
\node[blue] at (-3.65,3.2) {$S(3,5)$} ;

\end{tikzpicture}
}
\subcaption{The 9-move corresponding to the triangle $135$ results in a Hamiltonian cycle using one less edge of weight 2. Note that after the swaps $S(1,3)$ and $S(5,1)$ are performed, the only 3-swap that can reconnect the three cycles into one, is $S(3,5)$, implying the existence of the edge $35$, and thereby of the triangle $135$.}
\label{fig:foundTriangle}
\end{minipage}
\caption{Illustration of the reduction (left) and of a potential solution (right).}
\label{fig:9opt}
\end{figure}

\paragraph*{A triangle in $H$ implies an improving 9-move for $(G,w,\mathcal C)$.}
Let $abc$ be a triangle in $H$.
In particular, all three swaps $S(a,b)$, $S(b,c)$, and $S(c,a)$ exist.
Performing these three 3-swaps results in a spanning union of (vertex-disjoint) cycles, whose total weight is $w(\mathcal C)-1$.
Indeed $S(a,b)$ is swap of weight $-1$, while $S(b,c)$, and $S(c,a)$ are both neutral.

We thus only need to show that the three swaps result in a connected graph (hence, Hamiltonian cycle of lighter weight).
By performing the 3-swap $S(a,b)$, we create three components: (1) one on a vertex set $K_{a,b}$ such that $J_a \subseteq K_{a,b} \subseteq I_a$, (2) one containing the scopes of vertices of the $B$-side to the right (lower part) of the scope of $b$, and (3) one containing the scopes of vertices of the $B$-side to the left (upper part) of the scope of $b$.
Then the swap $S(c,a)$ glues (1) and (2) together, but also disconnects (4) a cycle on a vertex set $K_{c,a}$ such that $J_c \subseteq K_{c,a} \subseteq I_c$.
At this point, there are three cycles: (3), (1)$+$(2), and (4).
It turns out that the 3-swap $S(b,c)$ deletes exactly one edge in each of these three cycles: $b\underline{c}$ in (4), $l^-(b,c)$ in (3), and $r^-(b,c)$ in (1)$+$(2).
Therefore, $S(b,c)$ reconnects these three components into one Hamiltonian cycle.

\paragraph*{An improving k-move for $(G,w,\mathcal C)$ with $k \leqslant 9$ implies a triangle in $H$.}
We assume that there is an improving $k$-move $\mathcal M = (E^-,E^+)$ for $(G,w,\mathcal C)$ with $k \leqslant 9$.
Being improving, the $k$-move has to contain at least one improving 3-swap of $S(A,B)$.
Let $S(a,b)$ be a 3-swap of $S(A,B)$ in $\mathcal M$ such that for every other (improving) 3-swap $S(a,b')$ in $\mathcal M$, the chord $\underline{a}b'$ is wider than $\underline{a}b$.   
Since $S(a,b)$ exists, it holds in particular that $ab \in E(H)$.
Performing $S(a,b)$ results in the union of three cycles: (1) on a vertex set $K_{a,b}$ with $J_a \subseteq K_{a,b} \subseteq I_a$, and cycles (2) and (3) as described in the previous paragraph.

By the choice of $b$, the only remaining swaps of $\mathcal M$ touching $K_{a,b}$ are in $S(C,A)$.
So $\mathcal M$ has to contain a neutral 3-swap $S(c,a)$ for some $c \in C$.
This implies that $ac \in E(H)$.
Performing this swap results in three cycles: (3), (1)$+$(2), and (4), as described above.
% Bart writes: I found the argument 'the budget is spent and there is more than one component left' not so convincing, and upon reformulating I realized that these three lines are actually redundant for the argumentation; the last part of the argument is valid even without having first realized that the last swap is of the form S(B,C). Hence I removed these 3 lines from the proof.
%
%At this stage, if $S(a,b)$ and $S(c,a)$ are \emph{not} the only two 3-swaps of $S(A,B) \cup S(C,A)$ in $\mathcal M$, then the budget of 9 moves is spent (and there is strictly more than one connected component in the ``cycle'').
%
%Therefore, the third and last 3-swap of $\mathcal M$ has to be in $S(B,C)$.
To reconnect all three components into one Hamiltonian cycle, the 3-swap has to delete exactly one edge in (3), (1)$+$(2), and (4).
The only 3-swap that does so is $S(b,c)$.
This finally implies that $bc \in E(H)$.
Thus $abc$ is a triangle in $H$.
\end{proof}

We obtain the following theorem as a direct consequence of the previous lemma.
\begin{theorem}\label{thm:9opt}
  \textsc{Subcubic} \probOPTDec{9} requires time:
  \begin{enumerate}[(1)]
  \item $n^{1+\delta - o(1)}$ for a fixed $\delta > 0$, under the triangle hypothesis, and 
  \item $n^{4/3 - o(1)}$, under the strong triangle hypothesis,
  \end{enumerate}
  in expectation, even in undirected graphs with edge weights in $\{1,2\}$.
\end{theorem}

If we use general integral weights and not just $\{1,2\}$, we can show a stronger lower bound, by reducing from \textsc{Negative Edge-Weighted Triangle}.
Again, we can assume that the instance is partitioned into three sets $A$, $B$, $C$, and we look for a triangle $abc$ such that $w'(ab)+w'(bc)+w'(ac) < 0$, where $w'$ gives an integral weight to each edge. 
A truly subcubic (in the number of vertices) algorithm for this problem would imply one for \textsc{All-Pairs Shortest Paths}, which would be considered a major breakthrough.
The assumption that such an algorithm is not possible is called the APSP hypothesis.

We only change the above construction in the weight of the edges $l^-(x,y)$.
Now each edge $l^-(x,y)$ gets weight $-w'(xy)$.
From a \textsc{Negative Edge-Weighted Triangle}-instance with $n$ vertices, we obtain an equivalent instance of \textsc{Subcubic} \probOPTDec{9} with $\Oh(n^2)$ vertices, in time $\Oh(n^2)$.
So we derive the following.
\begin{theorem}\label{thm:9opt}
  \textsc{Subcubic} \probOPTDec{9} requires time $n^{3/2 - o(1)}$, under the APSP hypothesis.
\end{theorem}

\section{Lower bound for varying $k$} \label{sec:xp-lowerbound}
The goal of this section is to prove that, under the Exponential Time Hypothesis, the exponent of~$n$ must grow near-linearly with the value of~$k$ in the running time of \probKOPTDec, even on graphs of degree at most three. We give a reduction from a variant of the parameterized subgraph isomorphism problem, for which such a lower bound is known.

\defparproblem{\probSubIso}
{An undirected host graph~$G$, an undirected pattern graph~$H$, and a function~$\psi \colon V(G) \to V(H)$.}
{$k := |E(H)|$.}
{Is there a function~$\phi \colon V(H) \to V(G)$ such that~$\psi(\phi(i)) = i$ for all~$i \in V(H)$, and~$\{\phi(i),\phi(j)\} \in E(G)$ for all~$\{i,j\} \in E(H)$?}

The function~$\psi$ indicates for each vertex of~$G$, which role it is allowed to play in the subgraph isomorphism. For~$i \in V(H)$ we denote by~$\psi^{-1}(i)$ the set~$\{v \in V(G) \mid \psi(v) = i\}$. Hence in \probSubIso, the vertex set of the input graph is partitioned into classes~$(\psi^{-1}(i))_{i \in V(H)}$ and the question is whether one can find an occurrence of the pattern~$H$ in~$G$ by picking the image of~$i \in V(H)$ from the set~$\psi^{-1}(i)$. An isomorphism~$\phi$ between~$H$ and a subgraph of~$G$ is said to \emph{respect the partition} of~$G$ if, for all~$i \in V(H)$, the image of~$i$ belongs to~$\psi^{-1}(i)$.

\begin{theorem}[{\cite[\S 6]{Marx10a}}] \label{theorem:marx:subiso}
If \probSubIso for pattern graphs~$H$ on~$k$ edges and $n$-vertex host graphs~$G$ can be solved in time~$f(k) \cdot n^{o(k / \log k)}$ for any function~$f$, then the Exponential Time Hypothesis is false.
\end{theorem}

\begin{figure}[t]
\begin{center}
\includegraphics{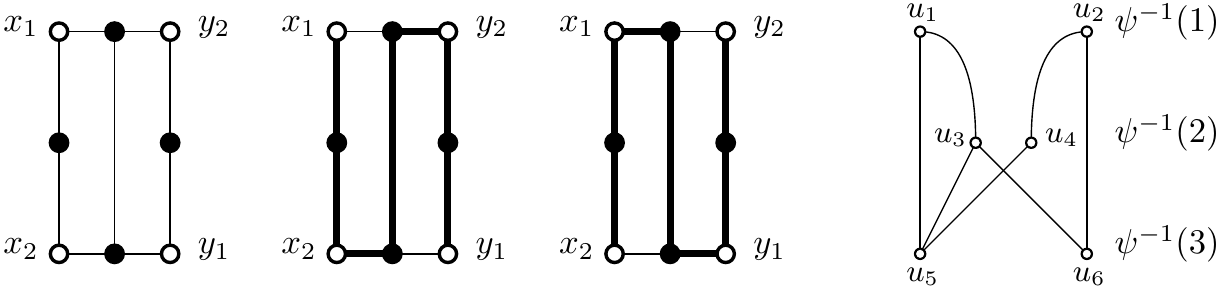}
\caption{Left: The domino gadget, its $(x_1,y_1)$-traversal, and its~$(x_2,y_2)$-traversal. Right: An input graph~$G$ of \probSubIso for~$H$ a triangle on vertex set~$\{1,2,3\}$ with~$k=3$ edges. The mapping~$\phi(1) = u_1, \phi(2) = u_3, \phi(3) = u_5$ is a partition-respecting solution.}\label{fig:domino}
\end{center}
\end{figure}

The following lemma gives the key construction for establishing lower bounds for finding \kOPT improvements. It is the main technical contribution of this section. Marx gave a related construction~\cite{Marx08} to prove that \probKOPTDec is W[1]-hard by a reduction from \textsc{$k$-Clique}, which was refined by Guo et al.~\cite{GuoHNS13} to provide a nearly-tight running time lower bound via \probSubIso. The next lemma refines the lower bound even further, showing it applies even for constant-degree graphs~$G$.

\newcommand{\rk}[0]{\ensuremath{{\rm rk}}}

\begin{lemma} \label{lemma:subiso:to:hamcycle}
There is a linear-time algorithm that, given a host graph~$G$, a pattern graph~$H$ with~$k$ edges, and a function~$\psi \colon V(G) \to V(H)$, outputs a graph~$G'$ of maximum degree~$3$, a weight function~$w \colon E(G) \to \{1,2\}$, an integer~$k' \in \Oh(k)$, and a Hamiltonian cycle~$\Cc^*$ in~$G$, such that:
\begin{enumerate}
	\item If there is a Hamiltonian cycle in~$G'$ whose weight is smaller than that of~$\Cc^*$, then there is a partition-respecting isomorphism between~$H$ and a subgraph of~$G$.\label{pty:cycle:implies:isomorphism}
	\item If there is a partition-respecting isomorphism between~$H$ and a subgraph of~$G$, then there is a Hamiltonian cycle in~$G'$ whose weight is smaller than that of~$\Cc^*$, which can be obtained by a \kprimeOPT move from~$\Cc^*$.\label{pty:isomorphism:implies:cycle}
\end{enumerate}
\end{lemma}
\begin{proof}
Consider the input~$(G,H,k,\psi)$. We may assume without loss of generality that~$H$ has no isolated vertices, since they are trivial to deal with. Hence~$\hat{k} := |V(H)| \leq 2k$. Identify the vertices of~$H$ with the integers~$\{1,\ldots,{\hat{k}}\}$, let~$n := |V(G)|$ and~$m := |E(G)|$. We assume without loss of generality that for each~$i \in V(H)$, for each~$v \in \psi^{-1}(i)$, vertex~$v$ is adjacent in~$G$ to at least one vertex of~$\psi^{-1}(j)$ for all~$j \in N_H(i)$. If this is not the case, then~$v$ can trivially not be used in an occurrence of the pattern subgraph and can safely be removed. Similarly, we assume without loss of generality that if~$\{u,v\} \in E(G)$, then~$\{\psi(u),\psi(v)\} \in E(H)$, as otherwise the edge can safely be removed. 

Order the vertices of~$G$ arbitrarily as~$u_1, \ldots, u_n$. Any subset of~$V(G)$ inherits a total ordering from the total ordering of~$V(G)$. For a subset~$U \subseteq V(G)$ and element~$u \in U$, we denote by~$\rk(u,U)$ the rank of~$u$ in the ordered set~$U$, i.e., the number of elements of~$U$ equal to~$u$ or preceding~$u$ in this total order. Hence the first element~$u'$ of~$U$ has~$\rk(u',U) = 1$ and the last element~$u''$ of~$U$ has~$\rk(u'',U) = |U|$.

All edges of~$G'$ have weight~$1$ unless stated otherwise. We employ three types of gadgets. Each gadget has a number of \emph{terminal vertices}, which are the only vertices of the gadget that will be adjacent to other vertices of~$G'$.
\begin{itemize}
	\item The domino gadget is depicted in Figure~\ref{fig:domino}. Its terminal vertices are~$x_1,y_1,x_2,y_2$. Observe that a Hamiltonian cycle traverses the domino gadget in exactly one of two ways: via an $x_1y_1$-path that visits all eight vertices, or via an $x_2y_2$-path that visits all eight vertices.
	\item A \emph{forcing gadget} consists of a path on three consecutive vertices~$(x, z, y)$, with terminals~$x$ and~$y$. If the gadget is contained in~$G'$, then the middle vertex~$z$ has degree two. Hence any Hamiltonian cycle of~$G'$ contains both edges incident on~$z$, and therefore contains exactly one non-gadget edge incident on each terminal.
	\item A \emph{choice gadget} of size~$\ell \geq 1$ consists of a cycle on~$3(\ell+1)$ vertices~$(x_0, y_0, z_0, x_1, y_1, z_1, \ldots, \linebreak[1] x_\ell, y_\ell, z_\ell)$, which represents a choice between~$\ell$ alternatives and a dummy choice for index~$0$. The vertices~$\{x_i, y_i \mid 0 \leq i \leq \ell\}$ are the terminals of the gadget. When a choice gadget is embedded in~$G'$, then all vertices~$z_i$ for~$0 \leq i \leq \ell$ have degree exactly two. Hence any Hamiltonian cycle of~$G'$ contains both edges incident on~$z_i$. When using the choice gadget in our construction, we assign weight~$2$ to all edges connecting from $x$-terminals of the gadget to vertices outside the gadget. We assign other edges incident on the gadget weight~$1$. This will enforce that sufficiently cheap Hamiltonian cycles enter the gadget \emph{once} at some terminal~$x_i$, then visit all its vertices through edges of weight~$1$ by going around the cycle, and exit the gadget at the corresponding~$y_i$.
\end{itemize}

Using these gadgets we now give the construction of~$G'$. To simplify the presentation, we use the following shorthand: when~$C$ is the identifier of a gadget that was inserted into~$G'$, then by~$C[x_i]$ (resp.~$C[y_i]$) we refer to the copy of terminal~$x_i$ (resp.~$y_i)$ in gadget~$C$. Refer to Figure~\ref{fig:subiso} for an illustration of the construction, which builds~$G'$ as follows:

\newcommand{\Fa}[0]{F^{\mathrm{\textsc{a}}}}
\newcommand{\Fe}[0]{F^{\mathrm{\textsc{e}}}}
\newcommand{\Fv}[0]{F^{\mathrm{\textsc{v}}}}

\begin{figure}[t]
\begin{center}
\includegraphics{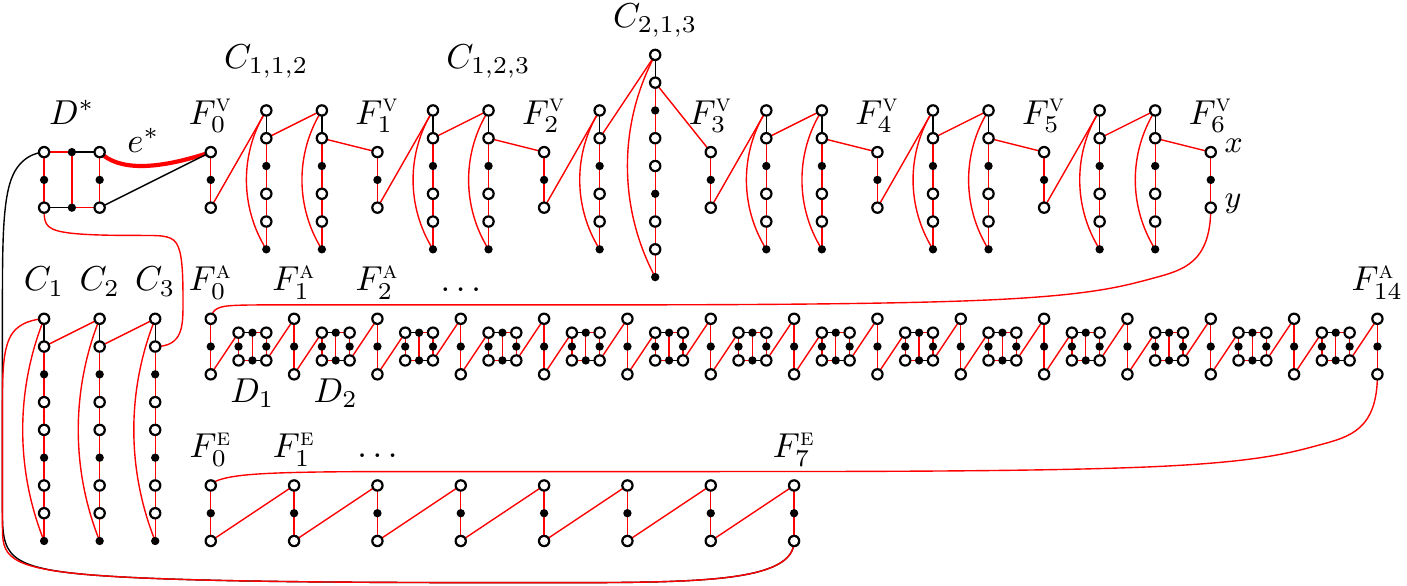}
\caption{Illustration of the reduction from \probSubIso to \probKOPTDec, applied to the input graph~$G$ from Figure~\ref{fig:domino}. Edges connecting domino gadgets $D_1, \ldots, D_{2m}$ to choice gadgets and to forcing gadgets~$\Fe_1, \ldots, \Fe_m$ have been omitted for clarity, as have the edges connecting choice gadgets~$C_1, \ldots, C_{\hat{k}}$ to~$\Fv_0, \ldots, \Fv_n$. The starting Hamiltonian cycle~$\Cc^*$ is highlighted in red; the special weight-2 edge~$e^*$ is bold.}\label{fig:subiso}
\end{center}
\end{figure}

% TO MAINTAIN: letters around i,j refer to vertices / edges of H, letters around uvw refer to vertices/edges of G, letters around xyz are terminals of gadgets and so only live in G'.

% TO MAINTAIN: In the default orientation from lower indices to higher indices, you enter a gadget at its x-terminal, then exit at its y-terminal.

\begin{enumerate}
	\item For each~$i \in [\hat{k}] = V(H)$, insert a choice gadget~$C_i$ of size~$|\psi^{-1}(i)|$ into~$G'$. This gadget controls which vertex of~$\psi^{-1}(i)$ is used in the solution. For each~$i \in [\hat{k} - 1]$, insert an edge between~$C_i[y_0]$ and~$C_{i+1}[x_0]$. Note that such edges have weight~$2$ since they are not internal to a choice gadget and are incident on an $x$-terminal of a choice gadget.\label{connections:successive:choice}
	\item Insert a forcing gadget~$\Fv_0$. For each~$\ell \in [n]$, do as follows for vertex~$u_\ell$. Let~$i := \psi(u_\ell)$ and~$p := \rk(u_\ell, \psi^{-1}(i))$.
	\begin{enumerate}
		\item Insert a forcing gadget~$\Fv_\ell = \Fv_{i,p}$ into~$G'$ to represent~$u_\ell$. We will use two different identifiers to refer to this forcing gadget:
		\begin{itemize}
			\item we refer to it as~$\Fv_\ell$ (since it is the forcing gadget for the~$\ell$-th vertex of~$G$), and
			\item we refer to it as~$\Fv_{i,p}$ (since it is the forcing gadget for the~$p$-th vertex from~$\psi^{-1}(i)$).
		\end{itemize}
		\item For each~$j \in N_H(i)$, insert a choice gadget~$C_{i,p,j}$ of size~$|N_G(u_\ell) \cap \psi^{-1}(j)|$. If~$u_\ell$ is used in a pattern subgraph, this choice gadget controls which edge between~$u_\ell$ and a vertex in~$\psi^{-1}(j)$ is used in the subgraph.
		\item Let~$N_H(i) = \{j_1, \ldots, j_{\deg_H(i)}\} \subseteq V(H) = [\hat{k}]$ be the neighbors of~$i$ in their natural order. We connect the gadgets as follows. \label{connections:Fv}
		\begin{itemize}
			\item Insert an edge between~$\Fv_{\ell-1}[y]$ and~$C_i[x_p]$. Insert an edge between~$C_i[y_p]$ and~$\Fv_\ell[x]$.\label{connections:to:ci}
			\item Insert an edge between~$\Fv_{\ell-1}[y]$ and~$C_{i,p,j_1}[x_0]$. 
			\item For each~$\alpha \in [\deg_H(i) - 1]$, insert an edge between~$C_{i,p,j_{\alpha}}[y_0]$ and~$C_{i,p,j_{\alpha+1}}[x_0]$. 
			\item Insert an edge between~$C_{i,p,j_{\deg_H(i)}}[y_0]$ and~$\Fv_\ell[x]$.
		\end{itemize}
	\end{enumerate}
	\item Insert a forcing gadget~$\Fa_0$. Let~$\vec{G}$ denote the directed graph obtained from~$G$ by replacing each edge~$\{u,v\}$ by two directed arcs~$(u,v)$ and~$(v,u)$. Let~$a_1, \ldots, a_{2m}$ be an arbitrary ordering of the arcs of~$\vec{G}$. 
	
	For each~$\ell \in [2m]$, we do the following. We have~$a_\ell = (u, v)$ with~$\{\psi(u), \psi(v)\} \in E(H)$. Let~$i := \psi(u)$ and~$p := \rk(u, \psi^{-1}(i))$. Let~$j := \psi(v)$ and let~$q := \rk(v, \psi^{-1}(j))$, so that arc~$a_\ell$ connects the $p$-th vertex of~$\psi^{-1}(i)$ to the $q$-th vertex of~$\psi^{-1}(j)$.
	\begin{enumerate}
		\item Insert a domino gadget~$D_\ell = D_{i,p,j,q}$ into~$G'$ for this arc~$a_\ell$.
		\item Insert a forcing gadget~$\Fa_\ell = \Fa_{i,p,j,q}$ into~$G'$ for this arc~$a_\ell$.
		\item We connect these gadgets as follows.
		\begin{enumerate}
			\item Insert an edge between~$\Fa_{\ell-1}[y]$ and~$D_\ell[x_1]$. Insert an edge between~$D_\ell[y_1]$ and~$\Fa_\ell[x]$.\label{connections:fa:domino}
			\item Since~$a_\ell = (u,v)$ is an arc of~$\vec{G}$, we have~$v \in N_G(u) \cap \psi^{-1}(j)$. Let~$r := \rk(v, N_G(u) \cap \psi^{-1}(j))$. Insert edges between~$\Fa_{\ell-1}[y]$ and~$C_{i,p,j}[x_r]$, and between~$C_{i,p,j}[y_r]$ and~$\Fa_\ell[x]$.\label{connections:fa:choice}
		\end{enumerate}
		This construction allows a Hamiltonian cycle to travel from~$\Fa_{\ell-1}[y]$ to~$\Fa_\ell[x]$ in two distinct ways: via the domino gadget~$D_\ell$, or via the choice gadget~$C_{i,p,j}$. The latter encodes that the undirected edge underlying~$a_\ell$ is used in the pattern subgraph.
	\end{enumerate}
	\item Insert a forcing gadget~$\Fe_0$. Let~$e_1, \ldots, e_m$ be an arbitrary ordering of the edges of the undirected graph~$G$. 
	
	For each~$\ell \in [m]$, we do the following. We have~$e_i = \{u, v\}$ for some~$u$ and~$v$ with~$\{\psi(u), \psi(v)\} \in E(H)$ with~$\psi(u) < \psi(v)$. That is, we only carry out this part of the construction for the orientation of the edge from smaller to higher endpoint, which is crucial. Let~$i := \psi(u)$ and~$p := \rk(u, \psi^{-1}(i))$. Let~$j := \psi(v)$ and let~$q := \rk(v, \psi^{-1}(j))$.
	\begin{enumerate}
		\item Insert a forcing gadget~$\Fe_\ell = \Fe_{i,p,j,q}$ into~$G'$.
		\item Insert an edge between~$\Fe_{\ell-1}[y]$ and~$\Fe_\ell[x]$.\label{connections:Fe:direct}
		\item Insert an edge between~$\Fe_{\ell-1}[y]$ and~$D_{i,p,j,q}[x_2]$.\label{connections:Fe:domino:first}
		\item Insert an edge between~$D_{i,p,j,q}[y_2]$ and~$D_{j,q,i,p}[x_2]$.
		\item Insert an edge between~$D_{j,q,i,p}[y_2]$ and~$\Fe_\ell[x]$.\label{connections:Fe:domino:last}
	\end{enumerate}
	This construction allows a Hamiltonian cycle to travel from~$\Fe_{\ell-1}[y]$ to~$\Fe_\ell[x]$ in two distinct ways: it can either go directly using the edge inserted in~\eqref{connections:Fe:direct}, or visit the domino for the~$(i,j)$ orientation of the edge (lower endpoint first), then visit the domino for the~$(j,i)$ orientation of the edge (higher endpoint first), and then continue to~$\Fe_\ell[x]$.
	\item Insert a final domino gadget~$D^*$, and connect as follows.
	\begin{enumerate}
		\item Insert an edge between~$\Fv_n[y]$ and~$\Fa_0[x]$. \label{connections:Fazero}
		\item Insert an edge between~$\Fa_{2m}[y]$ and~$\Fe_0[x]$. \label{connections:Fezero}
		\item Insert an edge between~$\Fe_m[y]$ and~$D^*[x_1]$. Insert an edge between~$D^*[y_1]$ and~$\Fv_0[x]$.\label{connections:Fem:Dstar}
		\item Insert an edge between~$\Fe_m[y]$ and~$C_1[x_0]$. Insert an edge between~$C_{\hat{k}}[y_0]$ and~$D^*[x_2]$. As the final and crucial step, we insert an \emph{edge~$e^*$ of weight~$2$} between~$D^*[y_2]$ and~$\Fv_0[x]$.\label{connections:Fe:via:choice}
	\end{enumerate}
\end{enumerate}

This concludes the construction of~$G'$. We set~$k' := 16k+4 \in \Oh(k)$. We can observe the following from the definition of~$G'$.

\begin{observation}
The maximum degree of~$G'$ is three.
\end{observation}

\begin{observation} \label{obs:weight:two}
The only edges of weight~$2$ in~$G'$ are those incident on the $x$-terminal of a choice gadget but not internal to that choice gadget, and the special edge~$e^*$. No edge of~$G'$ connects two $x$-terminals of distinct choice gadgets.
\end{observation}

\begin{observation} \label{obs:choicegadgets}
The total number of choice gadgets in~$G'$ is~$L := \hat{k} + \sum _{v \in V(G)} \deg_H(\psi(v))$.
\end{observation}

The idea behind the construction, and the special role of~$e^*$, is as follows. Consider a Hamiltonian cycle~$\Cc'$ in~$G'$, oriented so that it starts at~$\Fv_0[x]$ and moves in the direction of~$\Fv_0[y]$ through the forcing gadget. We will show that~$\Cc'$ visits the forcing gadgets in increasing order of index, first visiting the gadgets~$\Fv$, then the gadgets~$\Fa$ and some of the associated dominos, and finally the gadgets~$\Fe$. If~$\Cc'$ has visited all choice gadgets and dominos during this process, then after arriving at~$\Fe_m[y]$ it can simply close the cycle by going to~$D^*[x_1]$, visiting the gadget, and leaving from~$D^*[y_1]$ back to its starting point~$\Fv_0[x]$. Such a Hamiltonian cycle is relatively cheap because it leaves~$D^*$ via an edge of weight~$1$. On the other hand, if there is a choice gadget~$C_i$ for~$i \in V(H)$ that has not been visited when the cycle reaches~$\Fe_m[y]$, then it must be the case that the cycle has not visited \emph{any} choice gadget~$C_i$ up to that point. Then it is forced to move from~$\Fe_m[y]$ to~$C_1[x_0]$, using the edges between $0$-indexed terminals of successive choice gadgets~$C_i$, and returning from~$C_{\hat{k}}[y_0]$ to~$D^*[x_2]$. However, in this case the cycle will have to traverse~$D^*$ and close the loop by using the weight-$2$ edge~$e^*$ from~$D^*[y_2]$ to~$\Fv_0[x]$. Effectively, this construction yields a trivial Hamiltonian cycle in~$G'$ (see Figure~\ref{fig:domino}) that traverses the forcing gadgets linearly and visits all gadgets~$C_i$ at the end, but which incurs an extra cost of one. If there is a partition-respecting subgraph isomorphism, then there is a strictly cheaper Hamiltonian cycle that visits choice gadgets earlier, based on which vertices and edges are used in the solution subgraph. 

We now formalize this intuition, starting with the construction of the Hamiltonian cycle~$\Cc^*$ in~$G'$. To simplify our description of~$\Cc^*$, we will use the following terminology. When the cycle enters a gadget at a terminal~$x_i$, by \emph{traversing the gadget} we mean visiting all vertices internal to the gadget by a Hamiltonian path which ends at the corresponding terminal~$y_i$. For each of the three types of gadget, it is easy to verify that such a path exists. In the remainder, we will sometimes say that a Hamiltonian cycle makes a \emph{valid traversal} of a gadget if its intersection with the gadget consists of a single path between two matching terminals~$x_i$ and~$y_i$. Cycle~$\Cc^*$ is constructed as follows.

\begin{enumerate}
	\item Cycle~$\Cc^*$ starts at vertex~$\Fv_0[x]$ and traverses~$\Fv_0$ to~$\Fv_0[y]$. For each~$\ell \in [n]$ in increasing order, we have~$\Fv_\ell = \Fv_{i,p}$ for some~$i \in V(H)$ and~$p \in [|\psi^{-1}(i)]$, and proceed as follows.
	\begin{enumerate}
		\item The cycle proceeds from~$\Fv_{\ell-1}[y]$ to~$C_{i,p,j_1}[x_0]$, visits gadget~$C_{i,p,j_1}$, and moves from $C_{i,p,j_1}[y_0]$ to~$C_{i,p,j_2}[x_0]$, continuing in this way until arriving at~$C_{i,p,j_{\deg_H(i)}}[y_0]$ from which the cycle moves to~$\Fv_\ell[x]$. This is possible using the edges added in step~\eqref{connections:Fv}.
	\end{enumerate}
	\item After arriving at~$\Fv_n[x]$, the cycle~$\Cc^*$ moves through the gadget to~$\Fv_n[y]$, and goes to~$\Fa_0[x]$ using the edge inserted in step~\eqref{connections:Fazero}. At this point,~$\Cc^*$ has visited all vertices of forcing gadgets~$\Fv$ and choice gadgets of the form~$C_{i,p,j}$. For each~$\ell \in [2m]$ in increasing order, the cycle visits~$\Fa_{\ell-1}$ and continues by going from~$\Fa_{\ell-1}[y]$ to~$D_\ell[x_1]$, visiting the gadget, and going from~$D_\ell[y_1]$ to~$\Fa_\ell[x]$ using the edges of step~\eqref{connections:fa:domino}.
	\item After arriving at~$\Fa_{2m}[x]$, the cycle~$\Cc^*$ moves through the gadget to~$\Fa_{2m}[y]$, and goes to~$\Fe_0[x]$ using the edge inserted in step~\eqref{connections:Fezero}. The cycle now simply visits the gadgets~$\Fe_1, \ldots, \Fe_m$ consecutively, moving via~$\Fe_m[x]$ to~$\Fe_m[y]$. Step~\eqref{connections:Fe:direct} has inserted the edges to facilitate this.
	\item After arriving at~$\Fe_m[y]$, the cycle~$\Cc^*$ goes to~$C_1[x_0]$ via the edge of step~\eqref{connections:Fe:via:choice}, visits~$C_1$ to end up at~$C_1[y_0]$, and uses edges between the $0$-terminals of successive choice gadgets inserted in step~\eqref{connections:successive:choice} to visit all choice gadgets of the form~$C_i$ to arrive at~$C_{\hat{k}}[y]$, from which the cycle goes to~$D^*[x_2]$ via the edge of step~\eqref{connections:Fe:via:choice}.
	\item The cycle visits~$D^*$, ending with the weight-2 edge~$e^*$ from~$D^*[y_2]$ to~$\Fv_0[x]$ of step~\eqref{connections:Fe:via:choice} to close the cycle.
\end{enumerate}

\begin{claim} \label{claim:weight:Cstar}
The weight of~$\Cc^*$ is~$\beta := |V(G')| + 1 + L$, with~$L$ as in Observation~\ref{obs:choicegadgets}.
\end{claim}
\begin{claimproof}
Since~$\Cc^*$ is a Hamiltonian cycle, it has exactly~$|V(G')|$ edges. All edges of~$G'$ have weight one, except for~$e^*$ and the external edge incident to the $x$-terminal of each choice gadget. By Observation~\ref{obs:weight:two}, no edge connects two $x$-terminals of distinct choice gadgets. Hence the weight-2 edges on~$\Cc^*$ are~$e^*$ together with the edges of~$\Cc^*$ used to enter choice gadgets via an $x$-terminal. By Observation~\ref{obs:weight:two}, each choice gadget in~$G^*$ contributes a unique weight-2 edge to~$\Cc^*$. Hence the weight of~$\Cc^*$ is~$|V(G')| + 1 + L$.
\end{claimproof}

\begin{claim}
If there is a partition-respecting isomorphism between~$H$ and a subgraph of~$G$, then there is a Hamiltonian cycle in~$G'$ whose weight is smaller than that of~$\Cc^*$. Moreover, such a Hamiltonian cycle can be obtained by a \kprimeOPT move from~$\Cc^*$.
\end{claim}
\begin{claimproof}
Let~$\phi \colon V(H) \to V(G)$ be a partition-respecting isomorphism. We first construct a Hamiltonian cycle~$\Cc_\phi$ in~$G$ whose weight is smaller than that of~$\Cc^*$, and then argue that its Hamming distance to~$\Cc^*$ is suitably bounded. The cycle~$\Cc_\phi$ is constructed as follows:

\begin{enumerate}
	\item Cycle~$\Cc_\phi$ starts at vertex~$\Fv_0[x]$ and traverses~$\Fv_0$ to end at~$\Fv_0[y]$. For each~$\ell \in [n]$ in increasing order, the cycle then proceeds according to the following case distinction.
	\begin{enumerate}
		\item If~$u_\ell$ is used in the pattern subgraph, i.e., if~$\phi(\psi(u_\ell)) = u_\ell$, then we do as follows. Let~$i := \psi(u_\ell)$ and~$p := \rk(u_\ell, \psi^{-1}(i))$. From~$\Fv_{\ell-1}[y]$, the cycle jumps to~$C_i[x_p]$ via the edges added in step~\eqref{connections:to:ci}, traverses gadget~$C_i$, and exits from~$C_i[y_p]$ to~$\Fv_\ell[x]$.
		\item If~$u_\ell$ is not used in the pattern subgraph, then the cycle proceeds from~$\Fv_{\ell-1}[y]$ to~$C_{i,p,j_1}[x_0]$, visits gadget~$C_{i,p,j_1}$, and moves from $C_{i,p,j_1}[y_0]$ to~$C_{i,p,j_2}[x_0]$, continuing in this way until arriving at~$C_{i,p,j_{\deg_H(i)}}[y_0]$ from which the cycle moves to~$\Fv_\ell[x]$. This is possible using the edges added in step~\eqref{connections:Fv}.
	\end{enumerate}
	\item After arriving at~$\Fv_n[x]$, the cycle~$\Cc_\phi$ moves through the gadget to~$\Fv_n[y]$, and goes to~$\Fa_0[x]$ using the edge inserted in step~\eqref{connections:Fazero}. At this point,~$\Cc_\phi$ has visited all vertices of forcing gadgets~$\Fv$, all choice gadgets of the form~$C_{i,p,j}$ when the $p$-th vertex of~$\psi^{-1}(i)$ is \emph{not} used in the pattern subgraph, and all choice gadgets~$C_i$ for~$i \in V(H)$.
	\item For each~$\ell \in [2m]$ in increasing order, the cycle visits~$\Fa_{\ell-1}$ and continues from~$\Fa_{\ell-1}[y]$ based on a case distinction. Let~$a_\ell = (u,v)$ and recall that~$\{\psi(u), \psi(v)\} \in E(H)$.
	\begin{enumerate}
		\item If the undirected edge~$\{u,v\}$ is used in the pattern subgraph, that is, when~$\phi(\psi(u)) = u$ and~$\phi(\psi(v)) = v$, then the cycle skips domino gadget~$D_\ell$ and does the following instead. Let~$p := \rk(u, \psi^{-1}(\psi(u)))$ and~$r := \rk(v, N_G(u) \cap \psi^{-1}(\psi(v)))$. From~$\Fa_{\ell-1}[y]$, the cycle continues to~$C_{\psi(u), p, \psi(v)}[x_r]$, visits the choice gadget, leaving from~$C_{\psi(u), p, \psi(v)}[y_r]$ to~$\Fa_\ell[x]$ via the edges of step~\eqref{connections:fa:choice}.
		\item If the undirected edge~$\{u,v\}$ is not used in the pattern subgraph, then the cycle continues from~$\Fa_{\ell-1}[y]$ to~$D_\ell[x_1]$, visiting the gadget, and going from~$D_\ell[y_1]$ to~$\Fa_\ell[x]$ using the edges of step~\eqref{connections:fa:domino}. 
	\end{enumerate}
	\item After arriving at~$\Fa_{2m}[x]$, the cycle~$\Cc_\phi$ moves through the gadget to~$\Fa_{2m}[y]$, and goes to~$\Fe_0[x]$ using the edge inserted in step~\eqref{connections:Fezero}. At this point, cycle~$\Cc_\phi$ has visited all choice gadgets, all forcing gadgets~$\Fv$ and~$\Fa$, and all dominos gadgets of the form~$D_{i,p,j}$ \emph{except} for the dominos corresponding to arcs that are orientations of edges of the pattern subgraph encoded by~$\phi$. 
	\item For each~$\ell \in [m]$ in increasing order, the cycle visits~$\Fe_{\ell-1}$ and continues from~$\Fe_{\ell-1}[y]$ based on a case distinction. Let~$e_\ell = \{u,v\}$, ordered so that~$\psi(u) < \psi(v)$. Let~$i := \psi(u)$ and~$p := \rk(u, \psi^{-1}(i))$. Let~$j := \psi(v)$ and let~$q := \rk(v, \psi^{-1}(j))$.
	\begin{enumerate}
		\item If edge~$\{u,v\}$ is used in the pattern subgraph, that is, when~$\phi(\psi(u)) = u$ and~$\phi(\psi(v)) = v$, then the cycle goes from~$\Fe_{\ell-1}[y]$ to~$D_{i,p,j,q}[x_2]$, traversing that gadget to~$D_{i,p,j,q}[y_2]$, moving to~$D_{j,q,i,p}[x_2]$, traversing the gadget, and moving from~$D_{j,q,i,p}[y_2]$ to~$\Fe_\ell[x]$. This is possible using the edges of steps~\eqref{connections:Fe:domino:first}--\eqref{connections:Fe:domino:last}.
		\item If edge~$\{u,v\}$ is not used in the pattern subgraph, then the cycle simply goes from~$\Fe_{\ell-1}[y]$ to~$\Fe_\ell[x]$ using step~\eqref{connections:Fe:direct}.
	\end{enumerate}
	Observe that this process visits the dominos for \emph{both} orientations~$(u,v), (v,u)$ of edges~$\{u,v\}$ used in the pattern subgraph. When arriving at~$\Fe_m[y]$, the only vertices the cycle has not visited yet belong to~$D^*$.
	\item After arriving at~$\Fe_m[y]$, the cycle~$\Cc_\phi$ goes to~$D^*[x_1]$, traverses~$D^*$ to~$D^*[y_1]$, closing the cycle with the edge to~$\Fv_0[x]$ which has \emph{weight one}. The edges exist by step~\eqref{connections:Fem:Dstar}.
\end{enumerate}

It is clear from the construction that~$\Cc_\phi$ is a Hamiltonian cycle in~$G'$. Let us now consider its weight. It consists of~$|V(G')|$ edges. Since~$\Cc_\phi$ does not use the edge~$e^*$, the only edges of weight two on~$\Cc_\phi$ are the edges used to enter choice gadgets via their $x$-terminal. Hence the weight of~$\Cc_\phi$ is~$|V(G')| + L$, which is strictly smaller than the weight of~$\Cc^*$ by Claim~\ref{claim:weight:Cstar}.

It remains to show that~$\Cc_\phi$ can be obtained from~$\Cc^*$ by a~$\kprimeOPT$ move, for some~$k' \in \Oh(k)$. To this end, it is important to note the following: while the size of choice gadgets can be arbitrarily large compared to~$k$, any pair of valid traversals of a choice gadget differs in only~$\Oh(1)$ edges. Each valid traversal misses exactly one edge of the cycle~$(x_0, y_0, z_0, \ldots, x_\ell, y_\ell, z_\ell)$, and there are only two external edges incident on the gadgets in which the traversals can differ. Since~$\Cc^*$ and~$\Cc_\phi$ both only perform valid traversals of choice gadgets, this implies the following: for each choice gadget in~$G'$, the number of edges incident to that gadget that lie on~$\Cc^*$ but not on~$\Cc_\phi$ is at most three.

A similar statement can be made for domino gadgets. Since the two traversals of domino gadgets differ only in two internal edges, and two edges of a Hamiltonian cycle are incident on terminals of such a gadget, we have the following: for any domino gadget~$D$ of~$G'$, the number of edges incident on~$D$ that lie on~$\Cc^*$ but not on~$\Cc_\phi$, is at most four.

Since~$\Cc^*$ and~$\Cc_\phi$ both traverse all forcing gadgets in exactly the same way, we can bound the number of edges that lie on~$\Cc^*$ but not~$\Cc_\phi$ by bounding the number of domino and choice gadgets for which the two cycles differ in their traversal of the gadget. From the construction of~$\Cc_\phi$, it follows that the only gadgets traversed differently by the two cycles are:
\begin{enumerate}
	\item The choice gadgets~$C_i$ for~$i \in V(H)$.
	\item The domino gadget~$D^*$.
	\item For each edge~$\{u,v\}$ used in the solution subgraph~$\phi$, ordered such that~$\psi(u) < \psi(v)$, i.e., for edges with~$\phi(\psi(u)) = u$ and~$\phi(\psi(v)) = v$ and~$\psi(u) < \psi(v)$, the two Hamiltonian cycles differ in the traversals of the dominos~$D_{i,p,j,q}$ and~$D_{j,q,p,j}$, where~$i := \psi(u)$,~$p := \rk(u, \psi^{-1}(u))$,~$j := \psi(v)$, and~$q := \rk(v, \psi^{-1}(v))$. 
	\item For each vertex~$u$ in the image of~$\phi$, we have~$\phi(\psi(u)) = u$. Letting~$i := \psi(u)$ and~$p := \rk(u, \psi^{-1}(i))$, the two Hamiltonian cycles differ in the traversals of all~$\deg_H(i)$ choice gadgets~$C_{i,p,j}$ for~$j \in N_H(i)$.
\end{enumerate}
Hence there are at most~$\hat{k} + \sum _{i \in V(H)} \deg_H(i) = \hat{k} + 2k$ choice gadgets in which the traversal differs between~$\Cc^*$ and~$\Cc_\phi$, and at most~$1 + 2k$ domino gadgets in which the traversal differs. Since all edges on~$\Cc^*$ that do not lie on~$\Cc_\phi$ are accounted for by edges incident to gadgets whose traversal changes, the number of edges on~$\Cc^*$ that do not lie on~$\Cc_\phi$ is bounded by~$3 \cdot (\hat{k} + 2k) + 4(1+2k) \leq 6k + 6k + 8k + 4 = 16k+4 = k'$. The cycle~$\Cc_\phi$ can be obtained from~$\Cc^*$ by removing at most~$k'$ edges and inserting at most~$k'$ replacement edges; it can be obtained by a \kprimeOPT move.
\end{claimproof}

We now consider the other direction, and prove that any Hamiltonian cycle in~$G'$ of weight smaller than that of~$\Cc^*$ yields a partition-respecting subgraph isomorphism, regardless of whether or not that cycle can be reached by a \kprimeOPT move from~$\Cc^*$.

\begin{claim}
If there is a Hamiltonian cycle in~$G'$ whose weight is smaller than that of~$\Cc^*$, then there is a partition-respecting isomorphism between~$H$ and a subgraph of~$G$.
\end{claim}
\begin{claimproof}
We start by analyzing the structure of Hamiltonian cycles in~$G'$, based on the inserted gadgets. For any choice gadget of some size~$\ell$, which is a cycle on~$3(\ell+1)$ vertices~$(x_0, y_0, z_0, \ldots, x_\ell, y_\ell, z_\ell)$, the~$z$-vertices have degree two in~$G'$. Hence any Hamiltonian cycle contains both edges incident on each $z$-vertex. Since a Hamiltonian cycle cannot use \emph{all} edges of the~$3(\ell+1)$ cycle (it would close the cycle too early, and therefore fail to reach some vertices of~$G'$), it must avoid at least one edge. By the previous argument, this is an edge not incident to a $z$-vertex, hence an edge of the form~$(x_i, y_i)$ for~$0 \leq i \leq \ell$. Since~$x_i$ has degree two in any Hamiltonian cycle, the cycle uses an external edge incident on~$x_i$, which has weight two by construction. By Observation~\ref{obs:weight:two}, this implies that for each choice gadget in the graph, any Hamiltonian cycle uses at least one distinct weight-$2$ edge. Hence the weight of any Hamiltonian cycle in~$G'$ is at least~$|V(G')| + L$. 

If a Hamiltonian cycle does not use a valid traversal of a choice gadget, then the cycle avoids two or more edges of the~$3(\ell+1)$ internal edges on the gadget, which means avoiding two distinct edges~$\{x_i, y_i\}$ and~$\{x_j, y_j\}$. But such a Hamiltonian cycle uses at least two distinct external weight-$2$ edges incident on that gadget, implying its total cost is strictly larger than~$|V(G')| + L$. It follows that a Hamiltonian cycle that is cheaper than~$\Cc^*$ uses a valid traversal of each choice gadget. 

Since the only Hamiltonian paths ending at terminals of the domino gadgets are the two paths shown in Figure~\ref{fig:domino}, it follows that any Hamiltonian cycle uses one of these two traversals in each domino gadget. 

Now let~$\Cc$ be a Hamiltonian cycle in~$G'$ of weight smaller than that of~$\Cc^*$. By the previous arguments~$\Cc$ has weight exactly~$|V(G')| + L$. Since each Hamiltonian cycle uses at least~$L$ weight-$2$ edges incident on choice gadgets, if~$\Cc$ has weight~$|V(G')| + L$ then it cannot use the weight-$2$ edge~$e^*$. Hence~$\Cc$ cannot enter choice gadget~$C_1$ at~$C_1[x_0]$, since it would imply leaving~$C_{\hat{k}}$ at~$C_{\hat{k}}[y_0]$ and entering~$D^*$ at~$D^*[x_2]$ and leaving at~$D^*[y_2]$ via~$e^*$. This implies that the traversal of~$C_1$ used is not from~$C_1[x_0]$ to~$C_1[y_0]$, which implies that~$\Cc$ cannot enter~$C_2$ at~$C_2[x_0]$, and so on. So for each~$i \in V(H)$ cycle~$\Cc$ uses a traversal of~$C_i$ corresponding to its~$(x_{p_i}, y_{p_i})$ terminal pair for some~$p_i \in [|\psi^{-1}(i)|]$. Now construct a mapping~$\phi \colon V(H) \to V(G)$ by setting~$\phi(i) = v_i$, where~$v_i \in \psi^{-1}(i)$ is the unique vertex such that~$\rk(v_i, \psi^{-1}(i)) = p_i$. By this definition we trivially satisfy~$\psi(\phi(i)) = i$. It remains to prove that for all edges~$\{i,j\} \in E(H)$ we have~$\{\phi(v_i), \phi(v_j)\} \in E(G)$.

For each~$i \in V(H)$, cycle~$\Cc$ traverses choice gadget~$C_i$ from terminal~$x_{p_i}$ to terminal~$y_{p_i}$. Hence~$\Cc$ uses the unique external edge incident to~$C_i[x_{p_i}]$, whose other endpoint is~$\Fv_{\rk(v_i, V(G))-1}[y]$. (Recall that all orders derive from the ordering~$V(G) = (u_1, \ldots, u_n)$, so that~$\rk(u_\ell, V(G)) = \ell$ for all~$\ell \in [n]$.) Therefore the edge from~$\Fv_{\rk(v_i, V(G))-1}[y]$ to~$C_{i,p_i,j_1}[x_0]$ is not used on~$\Cc$, where~$N_H(i) = \{j_1, \ldots, j_{\deg_H(i)}\}$. As a consequence, cycle~$\Cc$ does not use the $(x_0,y_0)$-traversal of any choice gadget~$C_{i,p_i,j}$ for~$j \in N_H(i)$. Hence for each~$i \in V(H)$, for each~$j \in N_H(i)$, cycle~$\Cc$ uses an external edge~$e_{i,p_i,j}$ incident on an $x$-terminal of~$C_{i,p_i, j}$ with index larger than~$0$. By step~\eqref{connections:fa:choice} of the construction, the other endpoint of~$e_{i,p_i,j}$ is of the form~$\Fa_{\ell-1}[y]$, where~$a_\ell$ is an arc of~$\vec{G}$ starting in~$\phi(i)$ whose underlying undirected edge connects~$\phi(i)$ in~$G$ to a vertex of~$\psi^{-1}(j)$. Consequently, the edge from~$\Fa_{\ell-1}[y]$ to~$D_\ell[x_1]$ is not used on~$\Cc$, which means that domino~$D_\ell$ is not visited by an~$(x_1,y_1)$-traversal, and so must be visited by an~$(x_2,y_2)$-traversal. Hence for each vertex~$i$ of~$H$, for each~$j \in N_H(i)$, there is exactly one domino~$D_\ell$ in~$G'$ that corresponds to an arc from~$\phi(i)$ to a neighbor in~$\psi^{-1}(j)$ for which~$D_\ell$ is visited by an~$(x_2, y_2)$-traversal of~$\Cc$. 

Conversely to the conclusion from the previous paragraph, let us argue that a domino~$D_\ell$ for an arc~$a_\ell = (u,v)$ can only be visited by an~$(x_2,y_2)$-traversal if~$\phi(\psi(u)) = u$, i.e., if vertex~$u$ was selected as an image in the solution subgraph. This follows from the fact that an~$(x_2,y_2)$-traversal of~$D_\ell$ implies that~$\Cc$ does not use the edge from~$\Fa_{\ell-1}[y]$ to~$D_\ell[x_1]$, and therefore~$\Cc$ continues from~$\Fa_{\ell-1}[y]$ by visiting the choice gadget~$C_{\psi(u), \rk(u, \psi^{-1}(\psi(u))), \psi(v)}$ instead. This is only possible if that choice gadget was not visited by an~$(x_0,y_0)$-traversal, which is only possible if~$\Cc$ moves from~$\Fv_{\rk(u, V(G))-1}[y]$ to choice gadget~$C_i$. As each such choice gadget can only be visited once, it follows that for each~$i \in V(H)$ there is a unique vertex~$p$ for which the choice gadgets~$C_{i,p,j}$ for~$j \in N_H(i)$ can be visited by traversals other than the one connecting~$(x_0,y_0)$. By our definition of~$\phi(i)$ based on the traversal of~$C_i$, which determines the value of~$p$, it follows that a domino~$D_\ell$ for~$a_\ell = (u,v)$ can indeed only be visited by an~$(x_2,y_2)$-traversal if~$\phi(\psi(u)) = u$.

We claim that for an arc~$a_\ell = (u,v)$, the domino~$D_\ell$ can only be visited by an~$(x_2, y_2)$-traversal if the domino~$D_{\ell'}$ for the reverse arc~$a_{\ell'} = (v,u)$ is \emph{also} visited by an~$(x_2, y_2)$-traversal. To see this, suppose that~$\phi(u) < \phi(v)$. Then the only external edge leaving~$D_\ell[y_2]$ leads to~$D_{\ell'}[x_2]$, and therefore to leave~$D_\ell[y_2]$ requires the use of the~$(x_2,y_2)$-traversal of~$D_{\ell'}$. Similarly, if~$\phi(v) < \phi(u)$ then the only way to reach~$D_\ell[x_2]$ from an external edge is to use the edge from~$D_{\ell'}[y_2]$, hence requiring an~$(x_2,y_2)$-traversal of the latter. 

Using the claims of the previous two paragraphs, we complete the argument. We have seen that for each~$i \in V(H)$ and~$j \in N_H(i)$, cycle~$\Cc$ selects the domino of exactly one oriented arc from~$\phi(i)$ to a neighbor in~$\psi^{-1}(j)$ to visit by an~$(x_2, y_2)$-traversal. As the preceding argument shows that such a traversal is only possible if the arc for the reverse orientation is also selected via an~$(x_2, y_2)$-traversal, while only arcs incident on vertices in the range of~$\phi$ can be selected in this way, it follows that for each~$\{i,j\} \in E(H)$ we have~$\{\phi(i), \phi(j)\} \in E(G)$. This concludes the proof.
\end{claimproof}

Observe that the construction can easily be performed in time that is linear in the number of vertices and edges of~$G$ and~$H$ combined. This concludes the proof of Lemma~\ref{lemma:subiso:to:hamcycle}.
\end{proof}

As a direct consequence of Lemma~\ref{lemma:subiso:to:hamcycle} and the known ETH lowerbound for \probSubIso, we obtain the following.

\begin{theorem} \label{thm:eth:lowerbound}
There is no function~$f$ for which \probKOPTDec on $n$-vertex graphs of maximum degree~$3$ with edge weights in~$\{1,2\}$ can be solved in time~$f(k) \cdot n^{o(k / \log k)}$, unless ETH fails.
\end{theorem}

We remark that the lower bound also holds for \emph{permissive} local search algorithms which output an improved Hamiltonian cycle of arbitrarily large Hamming distance to the starting cycle~$\Cc$, if a cheaper cycle exists in the \kOPT neighborhood of~$\Cc$.

\bibliography{bounded-degree-k-opt}

\pagebreak
\appendix

%\section{An algorithm of running time $\Oh(n^{(0.1704+\epsilon_k)k})$}
%\label{sec:app-xp}
%\input{app-xp}

%\section{Omitted proofs from Section~\ref{sec:small-k-algorithms}}
%\label{sec:lemmas-small-k}
%\input{app-small-k}

%\section{A quasi-linear-time algorithm for $k=8$ and small weights}
%\label{sec:app-8}
%\input{app-8}

% \section{Lower bound under ETH} \label{appendix:eth}
% \input{ETH-lowerbound}

\section{Program to check Lemma~\ref{lem:relax2}}
\label{sec:app-code}
\begin{lstlisting}[language=Python,showstringspaces=false,columns=fullflexible]
#!/usr/bin/env python

k = 8

# Recursively enumerate all the connection patterns of k-moves and push them into patterns.
# M is the perfect matching on [2k], and M[i] is the mate of vertex i.
# -1 is a dummy value which means that the mate is not decided yet.
def enumerate_patterns(i, M, patterns):
    if i == 2 * k:
        if is_feasible(M):
            patterns.append(M[:])
    elif M[i] == -1:
        for j in range(i + 1, 2 * k):
            if M[j] == -1:
                (M[i], M[j]) = (j, i)
                enumerate_patterns(i + 1, M, patterns)
                M[i] = M[j] = -1
    else:
        enumerate_patterns(i + 1, M, patterns)

# Check the feasibility of the connection pattern M.
# Starting from vertex 0, we walk along the cycle and return whether the length is k.
def is_feasible(M):
    p = 0 # current vertex
    length = 0
    while p != 2 * k:
        length += 1
        p = M[p]
        if p % 2 == 0:
            p -= 1
        else:
            p += 1
    return length == k

# Decompose M into sequential swaps.
# The returned list [S_1, S_2, ..., S_c] is a partition of [k] such that each M[S_i] is a sequential swap.
# The list is sorted by |S_i|.
def sequential_swaps(M):
    seqs = []
    visited = [False] * (2 * k)
    for s in range(k):
        if not visited[s * 2]:
            # walk along the alternating cycle from 2s
            S = []
            p = s * 2
            while not visited[p]:
                S.append(p // 2)
                visited[p] = True
                p ^= 1
                visited[p] = True
                p = M[p]
            seqs.append(S)
    seqs.sort(key = len)
    return seqs

# Check whether M can be decomposed into two moves.
def is_reducible(M):
    seqs = sequential_swaps(M)
    c = len(seqs) # the number of sequential swaps
    for p in range(1, (1 << c) - 1): # try all the bipartitions of [c]
        M1 = M[:] # the swap consisting of sequential swaps not in p
        M2 = M[:] # the swap consisting of sequential swaps in p
        for i in range(c):
            if p >> i & 1:
                # delete the i-th sequential move from M1
                for j in seqs[i]:
                    (M1[j * 2], M1[j * 2 + 1]) = (j * 2 + 1, j * 2)
            else:
                # delete the i-th sequential move from M2
                for j in seqs[i]:
                    (M2[j * 2], M2[j * 2 + 1]) = (j * 2 + 1, j * 2)
        # If both the swaps are feasible, M is reducible.
        if is_feasible(M1) and is_feasible(M2):
            return True
    return False

# Swap i and j.
def swap(M, i, j):
    for p in range(2):
        mi = M[2 * i + p]
        mj = M[2 * j + p]
        (M[2 * i + p], M[mj]) = (mj, 2 * i + p)
        (M[2 * j + p], M[mi]) = (mi, 2 * j + p)


print("Enumerating connection patterns...")
patterns = []
enumerate_patterns(0, [-1] * (2 * k), patterns)
print("The number of connection patterns of {}-moves is {}.".format(k, len(patterns)))
patterns2 = []
for M in patterns:
    seqs = sequential_swaps(M)
    ks = list(map(len, seqs))
    if not ks == [2, 3, 3] or is_reducible(M):
        continue
    (i, j) = seqs[0]
    if i - 1 in seqs[1] and i + 1 in seqs[1] and j - 1 in seqs[2] and j + 1 in seqs[2]:
        patterns2.append(M)
    elif i - 1 in seqs[2] and i + 1 in seqs[2] and j - 1 in seqs[1] and j + 1 in seqs[1]:
        patterns2.append(M)
print("Among them, {} patterns meet the precondition in the lemma.".format(len(patterns2)))
for M in patterns2:
    seqs = sequential_swaps(M)
    ok = False
    for a in seqs[0]:
        if a - 1 in seqs[1]:
            A = 1
        else:
            A = 2
        if a - 2 in seqs[A] or a + 2 in seqs[A]:
            continue
        M1 = M[:]
        swap(M1, a, a + 1)
        M2 = M[:]
        swap(M2, a, a - 1)
        if (is_feasible(M1) or is_reducible(M1)) and (is_feasible(M2) or is_reducible(M2)):
            ok = True
    if not ok:
        print("counter example: {}".format(M))
        exit()
print("All of them satisfy the statement in the lemma.")
\end{lstlisting}

\end{document}